\documentclass[12pt]{iopart}
\usepackage{iopams}
\usepackage[numbers]{natbib}

\usepackage{xcolor,graphicx}

\expandafter\let\csname equation*\endcsname\relax
\expandafter\let\csname endequation*\endcsname\relax

\usepackage{amsmath,amssymb,amsthm}
\usepackage{hyperref}

\newtheorem{definition}{Definition}
\newtheorem{remark}{Remark}
\newtheorem{theorem}{Theorem}
\newtheorem{lemma}{Lemma}

\newcommand{\real}{\mathbb{R}}
\newcommand{\complex}{\mathbb{C}}
\newcommand{\cL}{\mathcal{L}}
\newcommand{\bx}{\mathbf{x}}
\newcommand{\bn}{\mathbf{n}}
\newcommand{\by}{\mathbf{y}}
\newcommand{\bh}{\mathbf{h}}
\newcommand{\bz}{\mathbf{z}}
\newcommand{\bD}{\mathbf{D}}

\newcommand{\bA}{\mathbf{A}}
\newcommand{\bM}{\mathbf{M}}
\newcommand{\bL}{\mathbf{L}}

\newcommand{\hphi}{\widehat{\phi}}
\newcommand{\hu}{\widehat{u}}
\newcommand{\hM}{\widehat{M}}
\newcommand{\bzero}{\mathbf{0}}
\newcommand{\ta}{\widetilde{a}}
\newcommand{\tb}{\widetilde{b}}
\newcommand{\ii}{\imath}
\newcommand{\Omegat}{\widetilde \Omega}

\newcommand{\supp}{\mathrm{supp}\text{ }}

\newcommand{\ttrue}{\text{true}}
\newcommand{\tmeas}{\text{meas}}
\renewcommand{\Re}{\text{Re}}
\renewcommand{\Im}{\text{Im}}
\DeclareMathOperator{\diag}{diag}

\begin{document}

\title[Restarted inverse Born series]{Restarted inverse Born series for
the\\ Schr\"odinger problem with\\ discrete internal measurements}
\author{Patrick Bardsley and Fernando Guevara Vasquez}

\address{Mathematics Department, University of Utah, 155 S 1400 E RM
233, Salt Lake City, UT 84112-0090, USA}
\eads{\mailto{bardsley@math.utah.edu}, \mailto{fguevara@math.utah.edu}}

\begin{abstract}
Convergence and stability results for the inverse Born
series [Moskow and Schotland, Inverse Problems, 24:065005, 2008] are
generalized to mappings between Banach spaces. We show that by
restarting the inverse Born series one obtains a class of iterative methods
containing the Gauss-Newton and Chebyshev-Halley methods. We use the
generalized inverse Born series results to show convergence of
the inverse Born series for the Schr\"odinger problem with discrete
internal measurements. In this problem, the Schr\"odinger potential is
to be recovered from a few measurements of solutions to the
Schr\"odinger equation resulting from a few different source terms. An
application of this method to a problem related to transient hydraulic tomography
is given, where the source terms model injection and measurement wells.
\end{abstract}

\ams{65N21, 35J10}
\submitto{\IP}
\maketitle

\section{Introduction}\label{sec:intro}
We consider the problem of finding a Schr\"odinger potential $q(\bx)$
(which may be complex) from discrete internal measurements of the solution
$u_i(\bx)$ to the Schr\"odinger equation 
\begin{equation}\label{eq:SchroEq}
\left\{
\begin{aligned}
	-\Delta u_i + qu_i &= \phi_i, &&\text{ for }\bx\in\Omega,\\
	u_i	&= 0, &&\text{ for }\bx\in\partial\Omega,
\end{aligned}
\right.
\end{equation}
in a closed bounded set $\Omega\subset\real^d$ for $d \geq 2$, and for
different (known) source terms $\phi_i \in C^\infty(\Omega)$,
$i=1,\ldots,N$. We further assume $q\in L^\infty(\Omega)$  is known in
$\Omega\backslash\Omegat$, where $\Omegat$ is a closed subset of
$\Omega$ with a finite distance separating $\partial\Omegat$ and
$\partial\Omega$.

The internal measurements we consider are  of the form
\begin{equation}\label{eq:sparsedata}
D_{i,j} = \int_\Omega \phi_j(\bx)u_i(\bx)d\bx, \text{for
$i,j=1,\ldots,N$}.
\end{equation}
The measurement $D_{i,j}$ is a weighted average of the field $u_i$
resulting from the $i-$th source term. Although it is not necessary for
our method to work, we assume for simplicity the same source terms are
used as weights for the averages.

A motivation for this inverse Schr\"odinger problem is transient hydraulic
tomography (see e.g. \cite{Cardiff:2011:THT} for a review). The
hydraulic pressure or head $v(\bx,t)$ in an underground reservoir or
aquifer $\Omega$ resulting from a source $\psi(\bx,t)$ (the injection
well) satisfies the initial value problem
\begin{equation}
\left\{
\begin{aligned}
Sv_t &= \nabla\cdot(\sigma\nabla v) - \psi, &&\text{ for }
\bx\in\Omega,~t>0,\\
   v(\bx,t) &= 0,  &&\text{ for } \bx\in\partial\Omega,~t>0,\\
   v(\bx,0) &= g(\bx), &&\text{ for } \bx\in\Omega.
\end{aligned}
\right.
\label{eq:ht}
\end{equation}
Here $S(\bx)$ is the storage coefficient and $\sigma(\bx)$ the hydraulic
conductivity of the aquifer. The inverse problem is to image both
$S(\bx)$ and $\sigma(\bx)$ from a series of measurements made by
fixing a source term at one well, and measuring the resulting pressure
response at the other wells. We show in section~\ref{sec:ht} that
the inverse problem of reconstructing $S(\bx)$ and $\sigma(\bx)$ from
these sparse (and discrete) internal pressure measurements, can be
recast as an inverse Schr\"odinger problem with discrete measurements as
in \eqref{eq:sparsedata}. 

The main tool we use here for solving the inverse Schr\"odinger problem
is inverse Born series. Inverse Born series have been used to solve
inverse problems in different contexts such as optical tomography
\cite{Markel:2003:IPOD,Markel:2007:CBS,Moskow:2009:NSI,Moskow:2008:CSI},
the Calder\'on or electrical impedance tomography problem \cite{Arridge:2012:IBS}
and in inverse scattering for the wave equation \cite{Kilgore:2012:IBS}. 

In section~\ref{sec:born} we generalize the inverse Born series
convergence results of \citet{Moskow:2008:CSI} and
\citet{Arridge:2012:IBS}, to nonlinear mappings between Banach spaces.
The convergence results of inverse Born series in this generalized
setting are given in section~\ref{sec:bornconv} and proved in
\ref{app:bornproof}, following the same pattern of the proofs
in \cite{Moskow:2008:CSI,Arridge:2012:IBS}.  This new framework is
applied in section~\ref{sec:ibex} to a few problems that have been
solved before with inverse Born series. We also show that both forward
and inverse Born series are closely related to Taylor series.  Since the
cost of calculating the $n-$th term in an inverse Born series grows
exponentially with $n$, we restart it after having computed a few $k$
terms (i.e. we truncate the series to $k$ terms and iterate). We show in
section~\ref{sec:iterated} that restarting the inverse Born series gives
a class of iterative methods that includes the Gauss-Newton and
Chebyshev-Halley methods. For the discrete measurements Schr\"odinger problem,
we prove that the necessary conditions for convergence of the inverse
Born series are satisfied (section~\ref{sec:sparsebounds}). Then in
section~\ref{sec:ht}, we explain how the transient hydraulic tomography
problem can be transformed into a discrete measurement Schr\"odinger
problem. Finally in section~\ref{sec:numerics} we present numerical
experiments comparing the performance of inverse Born series with other
iterative methods and their effectiveness for reconstructing the
Schr\"odinger potential in \eqref{eq:SchroEq} and for solving the transient
hydraulic tomography problem. We conclude in
section~\ref{sec:discussion} with a summary of our main results.

\section{Forward and inverse Born series in Banach spaces}
\label{sec:born}

We start by extending the notion of Born series and inverse Born series
\cite{Markel:2003:IPOD,Moskow:2008:CSI} to operators between Banach
spaces. The idea being to give a common framework for the convergence
proofs of the inverse Born series for diffuse waves
\cite{Moskow:2008:CSI}, the Calder\'on problem \cite{Arridge:2012:IBS}
and the discrete internal measurements Schr\"odinger problem. This generalization
also highlights that the inverse Born series are a systematic way of
finding non-linear approximate inverses for non-linear mappings. The
resulting approximate inverses are valid locally and have guaranteed
error estimates.

In sections~\ref{sec:fwdborn} and \ref{sec:invborn} we define forward
and inverse Born series for a mapping $f$ from a Banach space $X$ (the
parameter space) to another Banach space $Y$ (the data space).  Then in
section~\ref{sec:bornconv} we state local convergence results for
inverse Born series in Banach spaces that are valid under mild
assumptions on the forward Born series. The proofs are included in the
\ref{app:bornproof} as they are patterned after the proofs in
\cite{Moskow:2008:CSI,Arridge:2012:IBS}.  Examples of forward and
inverse Born series are included in section~\ref{sec:ibex}.

\subsection{Forward Born series}
\label{sec:fwdborn}
Let $X$ and $Y$ be Banach spaces and consider a mapping $f: X \to Y$. In
inverse problems applications $X$ is typically the parameter space and $Y$ the
data or measurements space. The forward problem is to find the measurements $y
= f(x)$ from known parameters $x$. The inverse problem is to estimate the
parameters $x$ knowing the measurements $y$.

Born series involve operators in $\cL ( X^{\otimes n} , Y)$, i.e.
bounded linear operators from $X^{\otimes n}$ to $Y$. Here tensor
products are used to define the Banach space 
\[ 
 X^{\otimes n} = \underbrace{X \otimes \cdots \otimes X}_{\text{$n$ times}},
\] 
endowed with the norm
 \begin{equation} \label{eq:tpnorm}
  \| x_1 \otimes \cdots \otimes x_n \|_{X^{\otimes n}} = \| x_1 \|_X \cdot \cdots \cdot \| x_n \|_X.
 \end{equation}
Notice that a map $a \in \cL ( X^{\otimes n} , Y)$ can be identified to a bounded multilinear (or $n-$linear) map $\ta: X^n \to Y$ defined by:
\[
 \ta(x_1,\ldots,x_n) = a(x_1 \otimes \cdots \otimes x_n).
\]

Forward Born series express the measurements for a parameter $x+h \in X$
near a known parameter $x \in X$, assuming knowledge of $y = f(x)$.

\begin{definition}\label{def:born}
A nonlinear map $f:X \to Y$ admits a Born series expansion at $x \in X$
if there are are bounded linear operators $a_n \in \cL( X^{\otimes n} ,
Y)$ (possibly depending on $x$) such that
\begin{equation}
 d(h) =  f(x + h) - f(x) = \sum_{n=1}^\infty a_n(h^{\otimes n}),
 \label{eq:born}
\end{equation}
and the $a_n$ satisfy the bound
\begin{equation}
 \|a_n\| \leq \alpha \mu^n ~ \text{for $n=0,1,\ldots$}.
 \label{eq:fwdest}
\end{equation}
\end{definition}
It follows from the bounds on the operators $a_n$, that the Born series
converges {\em locally}, i.e. when $h$ is sufficiently small:
\begin{equation}\label{eq:smallness}
 \| h \| <  1/\mu.
\end{equation}

This restriction on the size of the perturbation $h$ can be
thought of as the radius of convergence of the expansion about the point
$x$. 

\subsection{Inverse Born series}
\label{sec:invborn}
The purpose of inverse Born series is to recover $h$ from knowing the
difference in measurements $d(h) = f(x+h)-f(x)$ from a (known) reference
combination of parameters $x$ and measurements $y = f(x)$.  The original
idea in \cite{Markel:2003:IPOD} is to write a power series of the data
$d$,
\begin{equation}
 g(d) = \sum_{n=1}^\infty b_n(d^{\otimes n}),
 \label{eq:iborn}
\end{equation}
involving the operators $b_n \in \cL(Y^{\otimes n},X)$, which are
obtained by  requiring (formally) that $g$ is the inverse of $d(h)$,
i.e. $g(d(h)) = h$. By equating operators $\cL(X^{\otimes n},Y)$ with
the same tensor power $n$, the operators $b_n$ need to satisfy:
\begin{equation}
 \begin{aligned}
 I & = b_1(a_1)\\
 0 & = b_1(a_2) + b_2(a_1 \otimes a_1)\\
 0 & = b_1(a_3) + b_2(a_1 \otimes a_2) + b_2(a_2 \otimes a_1) + 
       b_3(a_1\otimes a_1 \otimes a_1)\\
       &\vdots\\
 0 & = \sum_{m=1}^n \sum_{s_1 + \cdots +s_m = n} b_m(a_{s_1} \otimes
 \cdots \otimes a_{s_m})
 \end{aligned}
\end{equation}
where $I$ is the identity in the parameter space $X$. The requirement
that $b_1 a_1 = I$ is quite strong and may not be possible, for example
when the measurement space $Y$ is finite dimensional and $X$ is infinite
dimensional. Nevertheless if we assume that $b_1$ is both a right and
left inverse of $a_1$ we can express the operators $b_n$ in terms of the
operators $a_n$ and $b_1$:
\begin{equation}\label{eq:iborn:coeff}
 \begin{aligned}
   b_2 &= - b_1 a_2 (b_1 \otimes b_1)\\
   b_3 &= - (b_1 a_3  + b_2(a_1 \otimes a_2) + b_2( a_2 \otimes a_1))
   ( b_1 \otimes b_1 \otimes b_1)\\
       &\vdots\\
   b_n &= - \left( \sum_{m=1}^{n-1} \sum_{s_1 + \cdots +s_m = n} b_m (a_{s_1} \otimes
 \cdots \otimes a_{s_m}) \right) ( b_1^{\otimes n} ).
 \end{aligned}
\end{equation}

Since an inverse of $a_1$ is not necessarily available, the key is to
choose $b_1 \in \cL(Y,X)$ as a regularized pseudoinverse of $a_1$ so
that $b_1 a_1$ is close to the identity, at least in some subspace. This
allows to define the inverse Born series.
\begin{definition}
Assume $f : X \to Y$ admits a Born series (Definition~\ref{def:born})
and let $b_1 \in \cL(Y, X)$. The {\em inverse Born series} for $f$ using
$b_1$ is the power series $g(d)$ given by \eqref{eq:iborn} where the
operators $b_n \in \cL (Y^{\otimes n},X)$ are defined for $n \geq 2$ by
\eqref{eq:iborn:coeff}.  Here again we note the dependence of the
operators $b_n$, $n \geq 2$, on the expansion point $x\in X$ and the
operator $b_1$.
\end{definition}

We now state results that guarantee convergence of the inverse Born
series, and give an error estimate between the limit of the inverse Born
series and the true parameter perturbation $h$. The error estimate
involves $\| ( I - b_1 a_1) h\|$, that is how well the operator $b_1
a_1$ approximates the identity for $h$. These results require that both
$h$ and $d(h) =  f(x+h) - f(x)$ are sufficiently small.

\subsection{Inverse Born series local convergence} \label{sec:bornconv}
Convergence and stability for the forward and inverse Born series were
established by \citet{Moskow:2008:CSI} for an inverse scattering problem
for diffuse waves (see also section~\ref{sec:dw}). Specifically they
obtained bounds on the operators $a_n$ in \eqref{eq:diffwaveborn}
similar to the bounds \eqref{eq:fwdest}. With these bounds, it is
possible to show convergence and stability of the inverse Born series
and even give a reconstruction error bound \cite{Moskow:2008:CSI}.

The convergence and stability proofs in \cite{Moskow:2008:CSI} for the
diffuse wave problem carry out without major modifications to the
general Banach space setting.  We give in this section a summary of results
analogous to those in \cite{Moskow:2008:CSI}. The proofs are deferred to
\ref{app:bornproof}, as they closely follow the proof pattern
in \cite{Moskow:2008:CSI}.

The following Lemma shows that if the forward Born operators satisfy the
bounds \eqref{eq:fwdest}, the operators $b_n$  are also bounded under
a smallness condition on the linear operator $b_1$ that is used to prime
the inverse Born series.

\begin{lemma} \label{lem:bj}
 Assume $f : X\to Y$  admits a Born series and that 
 \begin{equation}
  \| b_1\| < \frac{1}{(1+\alpha) \mu},
 \end{equation}
 where $\alpha$ and $\mu$ are as in Definition \eqref{def:born}. Then
 the coefficients \eqref{eq:iborn:coeff} of the inverse Born series satisfy the
 estimate
 \begin{equation}
  \| b_n \| \leq   \beta (  (1+\alpha)\mu \| b_1\|  )^n  , ~\text{for $n \geq 2$}
 \label{eq:invest}
 \end{equation}
 where 
 \begin{equation}
  \beta = \| b_1\| \exp\left( \frac{1}{1 - (1+\alpha)\mu \| b_1\|} \right).
 \end{equation}
\end{lemma}

Convergence of the inverse Born series follows from the bounds in
Lemma~\ref{lem:bj} and a smallness condition on the data $d$.
\begin{theorem}[Convergence of inverse Born series]\label{thm:smallness:inv}
 The inverse Born series \eqref{eq:iborn} induced by $b_1$ and associated with
 the forward Born series \eqref{eq:born} converges if
 \begin{equation}
  \| b_1\| < \frac{1}{(1+\alpha) \mu} 
  \label{eq:smallness:inv}
 \end{equation}
 and the data is sufficiently small
 \begin{equation}
 \| d \| < \frac{1}{(1+\alpha)\mu \| b_1\|}.
 \end{equation}
 If $h_*$ is the limit of the series, one can estimate the error due to
 truncating the series by
 \[
   \left \| h_* - \sum_{n=1}^N b_n(d^{\otimes n}) \right \| \leq  \beta  \frac{((1+\alpha)\mu \| b_1 \| \| d \|)^{N+1}}{1 - (1+\alpha)\mu \| b_1 \| \| d \|}.
 \]
\end{theorem}

Stability also follows using essentially the same proof as in \cite{Moskow:2008:CSI}.
\begin{theorem}[Stability of inverse Born series]\label{thm:stability}
 Assume $\|b_1\| < ( (1+\alpha)\mu)^{-1}$ and that we have two data $d_1$ and $d_2$ satisfying $M=\max(\|d_1\|,\|d_2\|) <  ((1+\alpha)\mu \| b_1\|)^{-1}$. Let $h_i = g(d_i)$ for $i=1,2$ (i.e. the limit of the inverse Born series). Then the reconstructions are stable with respect to perturbations in the data in the sense that:
 \begin{equation}
  \| h_1 - h_2 \| < C \| d_1 - d_2 \|,
 \end{equation}
 where the constant $C$ depends on $M$, $\alpha$, $\mu$, and $\|b_1\|$.
\end{theorem}

Theorem~\ref{thm:smallness:inv} guarantees
convergence of the forward and inverse Born series:
 \begin{equation}
  d = \sum_{j=1}^\infty a_j(h^{ j})
  \quad \text{and} \quad
  h_* = \sum_{j=1}^\infty b_j ( d^{ j} ).
  \label{eq:bornconv}
 \end{equation}
The limit $h_*$ of the inverse Born series is, in general, different from
the true parameter perturbation $h$.  The following theorem provides an
estimate of the error $\| h - h_* \|$. 

\begin{theorem}[Error estimate]\label{thm:error}
 Assuming that $\| h \| \leq M$, $\| b_1 a_1 h \| \leq M$ with $$M < \frac{1}{(1+\alpha) \mu},$$ and that the hypothesis of theorem~\ref{thm:smallness:inv} hold, i.e.  
 \[
  \| b_1 \| \leq \frac{1}{(1+\alpha) \mu} ~\text{and} ~ \| d \| \leq \frac{1}{(1+\alpha)\mu \| b_1 \| },
 \] 
 we have the following error estimate for the reconstruction error of the inverse Born series:
 \begin{equation} \label{eq:ib:error}
 \left\| h - \sum_{n=1}^\infty b_n (d^{n}) \right\| \leq C \| (I - b_1 a_1) h \|,
 \end{equation}
 where the constant $C$ depends only on $M$, $\alpha$, $\beta$ and $\mu$ and $\|b_1\|$.
\end{theorem}
The proofs of lemma~\ref{lem:bj}, theorems~\ref{thm:smallness:inv},
\ref{thm:stability}, and \ref{thm:error} can be found in \ref{app:bornproof}.

\section{Examples of forward and inverse Born series}
\label{sec:ibex}

We write examples of forward and inverse Born series in the framework of
section~\ref{sec:born}. We start by showing in section~\ref{sec:taylor}
that forward and inverse Born series are intimately related to Taylor
series. Another example is that of Neumann series
(section~\ref{sec:neumann}). We also include the forward and inverse
Born series from \cite{Moskow:2008:CSI,Arridge:2012:IBS}, namely those
for the diffuse waves for optical tomography (section~\ref{sec:dw}) and
the electrical impedance tomography problem (section~\ref{sec:eit}). We
finish the examples with the discrete internal measurements
Schr\"odinger problem (section~\ref{sec:qsparse}), which is the main
application of inverse Born series that we are concerned with here.

\subsection{Taylor series}
\label{sec:taylor}
\begin{itemize}
 \item[] {\bf Parameter space:} $X=$ Banach space
 \item[] {\bf Measurement space:} $Y=X$ (for simplicity)
 \item[] {\bf Forward map:} $f$ analytic (see e.g.
 \cite{Whittlesey:1965:AFB})
 \item[] {\bf Forward Born series coefficients:} About $x\in X$, the
 coefficients $a_n$ can be any operators in $\cL(X^{\otimes n},X)$
 agreeing with $f^{(n)}(x)/n!$ on the diagonal i.e. for any $h \in X$,
 \[
  a_n(h^{\otimes n}) = \frac{1}{n!}f^{(n)}(x)(h^{\otimes n}).
 \]
 Here $f^{(n)}$ is the $n-$th Fr\'echet derivative of $f$, see e.g. \cite[\S
 4.5]{Zeidler:1985:NFA} for a definition.
\end{itemize}

Here we use the theory of analytic functions between Banach spaces (see e.g.
\cite{Whittlesey:1965:AFB}) which assumes that the function $f$ is
$C^\infty$ and that the Taylor series of the function
\begin{equation}
 f(x+h) = \sum_{n=0}^\infty \frac{1}{n!}f^{(n)}(x) ( h^{\otimes n})
\end{equation}
converges absolutely and uniformly for $h$ small enough.  If in
addition we assume that $f$ admits a Born series expansion at $x$, then
we have
\[
d(h) = f(x+h) - f(x) = \sum_{n=1}^\infty
\frac{1}{n!}f^{(n)}(x)(h^{\otimes n}) = \sum_{n=1}^\infty a_n (h^{\otimes
n}).
\]
That is the Taylor series and Born series coefficients,
$f^{(n)}(x)/n!$ and $a_n$ respectively, agree at the diagonal
$h^{\otimes n}$. 

Since $f$ is $C^\infty$, the Fr\'echet derivatives $f^{(n)}$ are
symmetric in the sense that for any permutation $\pi$ of
$\{1,\ldots,n\}$ we have that
\[
f^{(n)}(h_1\otimes\cdots\otimes h_n) =
f^{(n)}(h_{\pi(1)}\otimes\cdots\otimes h_{\pi(n)}).
\]
The Born series coefficients $a_n$ in general {\em do not} satisfy this
property, however we can consider their symmetrization
$\ta_n:X^{\otimes n}\to Y$ defined by
\begin{equation}
\ta_n(h_1\otimes\cdots\otimes h_n) = \frac{1}{n!}\sum_{\pi}
a_n(h_{\pi(1)}\otimes\cdots\otimes h_{\pi(n)})
\label{eq:symm}
\end{equation}
where the summation is taken over all permutations $\pi$ of
$\{1,\ldots,n\}$. 

Clearly we have that
\[
\ta_n(h^{\otimes n}) = \frac{1}{n!}\sum_{\pi} a_n(h^{\otimes n}) =
a_n(h^{\otimes n}),
\]
and so we have the following equality:
\[
d(h) = f(x+h) - f(x) = \sum_{n=1}^{\infty}
\frac{1}{n!}f^{(n)}(x)(h^{\otimes n}) = \sum_{n=1}^\infty \ta_n(h^{\otimes
n}).
\]
We then have two analytic functions that are equal for $h$ sufficiently
small, therefore the symmetric operators $\frac{1}{n!}f^{(n)}(x)$ and $\ta_n$
must be identical (see \cite{Whittlesey:1965:AFB}). Therefore the Born
series and Taylor series coefficients are essentially the same, up to a
symmetrization.

If $a_1=f^{(1)}(x)$ is invertible (this is where the assumption $X=Y$ is
used), we can apply the implicit function theorem (see e.g.
\cite{Whittlesey:1965:AFB} or \cite[\S 4.6]{Zeidler:1985:NFA}) to
guarantee the existence of $f^{-1}$ in a neighborhood of $x$. Moreover
the inverse is analytic \cite{Whittlesey:1965:AFB} in a neighborhood of
$y = f(x)$ and admits a Taylor series near $y$ 
\begin{equation}
\label{eq:taylorinvf} f^{-1}(y+d) = \sum_{n=0}^\infty \frac{1}{n!}
(f^{-1})^{(n)}(y) (d^{\otimes n}).
\end{equation}

On the other hand, if $b_1 = a_1^{-1}$ we can define an inverse Born
series for $f$ as in \eqref{eq:iborn}.  By the error estimate for the
inverse Born series (Theorem~\ref{thm:error}) we can guarantee that $h =
g(d(h)) = g(f(x+h) - f(x))$ for $h$ and $d(h)$ sufficiently small. Since
$f$ is invertible in a neighborhood of $y$ we can also write $g$ in
terms $f^{-1}$
\[
 g(d) = f^{-1}(y+d) - f^{-1}(y) = f^{-1}(y+d) - x.
\]
Using the Taylor series \eqref{eq:taylorinvf} for $f^{-1}$  we can write
\begin{equation} \label{eq:itaylor}
 g(d) = \sum_{n=1}^\infty b_n(d^{\otimes n}) = \sum_{n=1}^\infty \frac{1}{n!} (f^{-1})^{(n)}(y)
 (d^{\otimes n}).
\end{equation}
As is the case for the forward Born operators $a_n$, the inverse Born
operators $b_n$ are in general not symmetric. If we consider their
symmetrization $\tb_n$ (as in \eqref{eq:symm}), then we find that the
symmetric operators $\tb_n$ and $\frac{1}{n!}(f^{-1})^{(n)}(y)$ are the
same.  Therefore inverse Born series is a way of calculating (up to a
symmetrization) the Taylor series for $f^{-1}$ from the Taylor series
for $f$.

\subsection{Neumann series}
\label{sec:neumann}
\begin{itemize}
 \item[] {\bf Parameter space:} $X = \real^N$
 \item[] {\bf Measurement space:} $Y = \real^{n\times n}$
 \item[] {\bf Forward map:} $f(\bx) = \bM^T ( \bL - \diag(\bx))^{-1} \bM$, where $\bL \in \real^{N\times N}$ is invertible and $\bM \in \real^{N\times n}$.
 \item[] {\bf Forward Born series coefficients:} About $\bzero$, the coefficients are $a_n(\bh) = \bM^T (\bL^{-1} \diag(\bh))^n \bL^{-1} \bM$.
\end{itemize}

The forward Born series in this is example comes from the Neumann series
for the inverse of $\bL - \diag(\bh)$, when it exists. Indeed if for
some matrix induced norm $\| \bL^{-1} \diag(\bh) \| < 1$, this inverse
exists and is given by the Neumann series 
\begin{equation}
 (\bL - \diag(\bh))^{-1} = \left( \sum_{n=0}^\infty (\bL^{-1} \diag(\bh))^n\right) \bL^{-1}.
\end{equation}
The forward Born series is then
\begin{equation}
 \begin{aligned}
 f(\bh) - f(\bzero) & =\bM^T ( \bL - \diag(\bh))^{-1} \bM  - \bM^T \bL^{-1} \bM\\
 &= \sum_{n=1}^\infty  \bM^T (\bL^{-1} \diag(\bh))^n \bL^{-1} \bM.
 \end{aligned}
\end{equation}

The inverse Born series can be defined by using as $b_1$ a regularized
pseudoinverse of the linear map $a_1(\bh) = \bM^T \bL^{-1} \diag(\bh)
\bL^{-1} \bM$. By the convergence results of section~\ref{sec:bornconv},
the inverse Born series converges under smallness conditions for $\bh$,
$f(\bh) - f(\bzero)$ and $b_1$.

This problem is motivated by a discretization of the Schr\"odinger
equation $\Delta u - q u = \phi$ with finite differences. The matrix
$\bL$ is the finite difference discretization of the Laplacian and $\bh$
is the Schr\"odinger potential at the discretization nodes. The matrix
$\bM$ corresponds to different source terms $\phi$, which are also used
to measure $u$ (collocated sources and receiver setup as the one we use
for the Schr\"odinger problem with discrete internal measurements in
section~\ref{sec:qsparse}). This example can be easily modified when the
discretization of the $q u$ term in the Schr\"odinger equation is not a
diagonal matrix (as is often the case for finite elements). The
collocated sources and receivers setup can be changed as well by using a
matrix other than $\bM^T$ in the definition of $f(\bx)$.

\subsection{Optical tomography with diffuse waves model \cite{Moskow:2008:CSI}}
\label{sec:dw}
In the diffuse waves approximation for optical tomography (see e.g.
\cite{Arridge:1999:OTM} for a review), the energy density $G_q(\bx,\by)$
resulting from a point source $\by \in \Omega$ satisfies a Schr\"odinger type
equation: 
\begin{equation}\label{eq:optommeas}
  \left\{
  \begin{aligned}
  -\Delta_\bx  G_q(\bx,\by) + q(\bx) G_q(\bx,\by) &= -\delta(\bx-\by), &&\text{for $\bx \in \Omega$,}\\
  G_q(\bx,\by) + \ell \bn(\bx) \cdot \nabla_\bx G_q(\bx,\by) &= 0, &&\text{for $\bx \in \partial \Omega$},
  \end{aligned}
  \right.
 \end{equation}
where the domain $\Omega \subset \real^d$, $d\geq 2$ has a smooth boundary
$\partial \Omega$, and $q(\bx) \geq 0$ is the absorption coefficient. The
$\ell \geq 0$ in the Robin boundary condition is given and, as usual,
$\bn(\bx)$ denotes the unit outward pointing normal vector to
$\partial\Omega$ at $\bx$. The inverse problem here is to recover the
absorption coefficient $q(\bx)$ from knowledge of $G_q(\bx,\by)$ on
$\partial\Omega \times \partial \Omega$. This data amounts to taking
measurements of the energy density at all $\bx \in \partial\Omega$ for
all source locations $\by \in \partial\Omega$ or to knowing the
Robin-to-Dirichlet map for $q$. If the difference between the
absorption coefficient $q(\bx)$ and a known reference coefficient
$q_0(\bx)$ is supported in some $\Omegat \subset \Omega$ (with $\partial
\Omega$ and $\partial \Omegat$ separated by a finite distance), then
$G_q$ satisfies the Lippmann-Schwinger type integral equation:
\begin{equation}
 G_q (\bx, \by)  = G_{q_0} (\bx,\by) + \int_{\Omegat} d\bz\; G_{q_0}(\bx,\bz)
 (q(\bz) - q_0(\bz)) G_q(\bz,\by).
 \label{eq:lippschwin}
\end{equation}

\citet{Moskow:2008:CSI} show that the forward Born or scattering series
for this problem can be defined as follows.

\begin{itemize}
 \item {\bf Parameter space:} $X=L^p(\Omegat)$ for $2\leq p \leq \infty$.
 \item {\bf Measurement space:} $Y=L^p (\partial \Omega \times \partial \Omega)$
 \item {\bf Forward map:} $f: q \to G_q(\bx,\by)|_{\partial \Omega
 \times \partial \Omega}$.
 \item {\bf Forward Born series coefficients:} For $\eta_1,\ldots,\eta_n \in
 L^p(\Omegat)$ and $\bx_1,\bx_2 \in \partial \Omega$, the coefficient
 for the Born series expansion about $q=q_0$ is
 \begin{multline} \label{eq:diffwaveborn}
  (a_n (\eta_1 \otimes \cdots \otimes \eta_n))(\bx_1,\bx_2) =\\ 
  \int_{\Omegat^n} G_{q_0}(\bx_1,\by_1) G_{q_0}(\by_1,\by_2) \ldots
  G_{q_0}(\by_{n-1},\by_n) G_{q_0}(\by_n,\bx_2)\\ \eta_1(\by_1) \ldots
  \eta_n(\by_n) ~ d\by_1 \ldots d\by_n.
 \end{multline}
\end{itemize}

In particular, the results of \citet{Moskow:2008:CSI} guarantee that the
operators $a_n$ satisfy the bounds \eqref{eq:fwdest} assuming $q_0$ is
constant and that $q$ is sufficiently close to $q_0$. Therefore one can
define an inverse Born series through the procedure
\eqref{eq:iborn:coeff}, and this series converges under appropriate
conditions (see \cite{Moskow:2008:CSI} and section~\ref{sec:bornconv}).

\subsection{The Calder\'on or electrical impedance tomography problem
\cite{Arridge:2012:IBS}}\label{sec:eit} The electric potential inside a
domain $\Omega$ with positive conductivity $\sigma(\bx) \in
L^\infty(\Omega)$ resulting from a point source located at $\by \in
\Omega$ satisfies the equation
\begin{equation}\label{eq:condeq}
\left\{
\begin{aligned}
 \nabla_\bx \cdot [ \sigma(\bx)\nabla_\bx G_\sigma(\bx,\by) ] &= -\delta(\bx-\by),
 &&\text{for $\bx\in\Omega$}\\
 G_\sigma(\bx,\by) + z \sigma\bn(\bx)\cdot\nabla_\bx G_\sigma(\bx,\by) &=
 0,&&\text{for $\bx\in\partial\Omega$.}
\end{aligned}
\right.
\end{equation}
Here we assume the contact impedance $z\geq0$ is known and that $\sigma$
is constant on $\partial \Omega$. The domain $\Omega$ is also assumed to
be in $\real^d$, $d\geq 2$ and with smooth boundary. The electric
impedance tomography (EIT) problem consists in recovering the
conductivity $\sigma$ from the Robin-to-Dirichlet map, i.e. from
knowledge of $G_\sigma(\bx,\by)$ on $\partial\Omega \times
\partial\Omega$ (see e.g. \cite{Borcea:2002:EIT} for a review of EIT).
If the difference between $\sigma$ and a known reference conductivity
$\sigma_0$ is supported in $\Omegat \subset \Omega$ (with $\partial
\Omegat$ at a finite distance from $\partial \Omega$), $G_\sigma$
satisfies the integral equation
\begin{equation}
 G_\sigma(\bx,\by) = G_{\sigma_0}(\bx,\by) + \int_{\Omegat} d\bz\;
 G_{\sigma_0}(\bx,\bz) \nabla_\bz
 \cdot[(\sigma(\bz)-\sigma_0(\bz))\nabla_\bz G_\sigma(\bz,\by) ].
\end{equation}
Integrating by parts and using that $\sigma = \sigma_0$ on $\partial
\Omega$, $G_\sigma$ obeys a Lippmann-Schwinger type equation:
\begin{equation}
 G_\sigma(\bx,\by) = G_{\sigma_0}(\bx,\by) - \int_{\Omegat} d\bz\;
 (\sigma(\bx)-\sigma_0(\bx))
 \nabla_\bz G_{\sigma_0}(\bx,\bz) 
 \cdot \nabla_\bz G_\sigma(\bz,\by).
\end{equation}
As shown by \citet{Arridge:2012:IBS}, one can then define a forward Born
series that can be summarized as follows.
\begin{itemize}
\item{\bf Parameter Space:} $X=L^\infty(\Omegat)$.
\item{\bf Measurement space:}
$Y=L^\infty(\partial\Omega\times\partial\Omega)$.
\item{\bf Forward map:} $f : \sigma \to G_\sigma(\bx,\by)|_{\partial
\Omega \times \partial \Omega}$.
\item{\bf Forward Born series coefficients:} For $\eta_1,\ldots,\eta_n\in
L^\infty(\Omegat)$ and $\bx_1,\bx_2 \in \partial \Omega$, the
coefficient for the Born series expansion about $\sigma=\sigma_0$
is
\begin{multline}\label{eq:caldbornop}
a_n (\eta_1 \otimes\cdots\otimes \eta_n)(\bx_1,\bx_2) =\\
(-1)^n\int_{\Omegat} d\by_1\; \eta_1(\by_1)\nabla_{\by_1}G_{\sigma_0}(\by_1,\bx_1)\cdot
\nabla_{\by_1}\int_{\Omegat} d\by_2\;\eta_2(\by_2)\nabla_{\by_2}G_{\sigma_0}(\by_2,\by_1)\cdot\\
	\cdots\nabla_{\by_{n-1}}\int_{\Omegat} d\by_n\;
	\eta_n(\by_n)\nabla_{\by_n}G_{\sigma_0}(\by_n,\by_{n-1})\cdot\nabla_{\by_n}G_{\sigma_0}(\by_n,\bx_2).
\end{multline}
\end{itemize}

\citet{Arridge:2012:IBS} show that for $\sigma_0$ constant, the
operators $a_n$ satisfy the bounds
\eqref{eq:fwdest} and so an inverse Born series can be defined following
the procedure \eqref{eq:iborn:coeff}. The convergence of this series is
established in \cite{Arridge:2012:IBS} and can also be shown using the
generalization in section~\ref{sec:bornconv}.

\subsection{The Schr\"odinger problem with discrete internal measurements}
\label{sec:qsparse}

Instead of having infinitely many measurements as in the optical
tomography inverse Schr\"odinger problem (outlined in
section~\ref{sec:dw}), we consider here the case where we only have
access to {\em finitely many} internal measurements $D_{i,j}$ (see
equation \eqref{eq:sparsedata}) of the fields $u_i$, $i=1,\ldots,N$,
satisfying \eqref{eq:SchroEq}. We also allow the Schr\"odinger potential
in \eqref{eq:SchroEq} to be complex (as discussed in
section~\ref{sec:ht}, this is useful when solving the
transient hydraulic tomography problem).

The Green function $G_q(\bx,\by)$ for the problem \eqref{eq:SchroEq}
satisfies \eqref{eq:optommeas} with homogeneous Dirichlet boundary
conditions (instead of homogeneous Robin boundary conditions).  The
fields $u_i$ can be expressed in terms of the Green function $G_q$ as
\begin{equation}
 u_i(\bx) = - \int_\Omega d\by \; G_q(\bx,\by)
 \phi_i(\by),~i=1,\ldots,N.
 \label{eq:repform}
\end{equation}
If the difference between the Schr\"odinger potential $q(\bx)$ and known
reference $q_0(\bx)$ is supported in $\Omegat \subset \Omega$ (with
$\partial \Omegat$ and $\partial \Omega$ separated by a finite
distance), $G_q$ and $G_{q_0}$ are still related by the
Lippmann-Schwinger type equation \eqref{eq:lippschwin}. By a fixed point
procedure we can define a forward Born series as follows.

\begin{itemize}
\item{\bf Parameter Space:} $X=L^\infty(\Omegat)$.
\item{\bf Measurement Space:} $Y = \complex^{N \times N}$, with norm $\|
\bA \| = \max_{i,j=1,\ldots N} |A_{i,j}|$.
\item{\bf Forward map:} Owing to \eqref{eq:repform},  the data $\bD$ in
\eqref{eq:sparsedata} becomes:
 \[
  f : q \to \bD = -\left[ \int_{\Omega^2} d\bx d\by \; \phi_i(\by) \phi_j(\bx)
  G_q(\bx,\by) \right]_{i,j=1\ldots N}.
 \]
 \item{\bf Forward Born series coefficients: } For $\eta_1,\ldots,\eta_n\in
L^\infty(\Omegat)$ the coefficient for the Born series expansion about
$q_0$ is 
\begin{multline}\label{eq:bornsparse}
[a_n(\eta_1\otimes\cdots\otimes \eta_n)]_{i,j} =\\ (-1)^n\int_{\Omegat^{n+2}}G_{q_0}(\bx,\by_1)G_{q_0}(\by_1,\by_2)\cdots
G_{q_0}(\by_{n-1},\by_n)G_{q_0}(\by_n,\bz)\cdot\\
	\eta_1(\by_1)\cdots \eta_n(\by_n) \phi_i(\bz)\phi_j(\bx) \; 
	d\bz d\by_1\cdots d\by_nd\bx,
\end{multline}
for $i,j=1,\ldots,N$. Note that we have have assumed $\supp \phi_i
\subset \Omegat$ so that instead of integrating over $\Omegat^n \times
\Omega^2$  integrate over $\Omegat^{n+2}$.
\end{itemize}

We show in section~\ref{sec:sparsebounds} that the operators $a_n$
satisfy the bounds \eqref{eq:fwdest} (with $q_0$ not necessarily
constant), so it is possible to show convergence of the corresponding
inverse Born series by the results of section~\ref{sec:bornconv}.

\section{Inverse Born series and iterative methods}\label{sec:iterated}

The main goal of this section is to show that inverse Born series can be
used to design superlinear\footnote{We recall that superlinear
convergence of $x_n$ to $x_*$ means that $\| x_{n+1} - x_* \| \leq \epsilon_n \| x_n - x_*
\|$, where $\epsilon_n \to 0$ as $n \to \infty$.}  iterative methods
converging to an approximation $x_*$ of the true parameter $x_{\ttrue}$ from
knowing measurements $y_{\tmeas} = f(x_{\ttrue})$ and the forward map $f: X \to Y$. The
iterative methods we study here are of the form
\[
 \left\{
 \begin{aligned}
 x_0 &= \text{given},\\
 x_{n+1} &= T_n(x_n),~\text{for $n\geq 0$,}
 \end{aligned}
 \right.
\]
where $T_n : X \to X$. Of course, for such an iterative method to be useful, the
iterates $x_n$ need to converge to $x_*$ as $n \to \infty$ (with an a priori
rate of convergence) and one should be able to estimate the error $\| x_{\ttrue}
- x_* \|$ between the desired parameter $x_{\ttrue}$ and the limit $x_*$.

\subsection{Inverse Born series as an iterative method}
We start by reformulating the results of section~\ref{sec:bornconv}
in the context of iterative methods. Let us assume that we have a good
guess $x_0$ for $x_{\ttrue}$, and that we know the forward Born
series about $x_0$, i.e. we know the coefficients $a_j[x_0] \in
\cL(X^{\otimes j}, Y)$ so that
\[
 f(x) - f(x_0) = \sum_{j=1}^\infty a_j[x_0] (x-x_0)^{\otimes j}.
\]
Theorem~\ref{thm:smallness:inv} means that for an appropriate choice of
$b_1[x_0]$, if $\|x_0-x_{\ttrue}\|$ and
$\|f(x_0)-y_{\tmeas}\|$ are sufficiently small then the inverse
Born series
\begin{equation}
 x_n - x_0 = \sum_{j=1}^n b_j[x_0] (y_{\tmeas}-f(x_0))^{\otimes j},
 \label{eq:iborn2}
\end{equation}
converges {\em linearly}\footnote{We recall that linear convergence rate
of $x_n$ to $x_*$ means that there is some $0<C<1$ such that $\|x_{n+1}
- x_*\| \leq C \| x_n - x^*\|$.} to some $x_* \in X$ as $n\to \infty$. Here we
write explicitly the dependence of the inverse Born operators $b_n[x_0]$
(defined recursively as in \eqref{eq:iborn:coeff}) on the reference
parameter $x_0$. Notice that the inverse Born series \eqref{eq:iborn2}
can be written as the iterative method,
\begin{equation}
 \left\{
 \begin{aligned}
 x_0 &= \text{given},\\
 x_{n+1} &= x_n + b_{n+1}[x_0] (y_{\tmeas}-f(x_0))^{\otimes n+1},~
 \text{for $n\geq 0$}.
 \end{aligned}
 \right.
 \label{eq:it:ibs}
\end{equation}
The error estimate of theorem~\ref{thm:error}
quantifies how close the limit $x_*$ of the iterative method
\eqref{eq:it:ibs} is to the true parameter $x_{\ttrue}$, i.e. there is some
$C>0$ such that 
\begin{equation}
 \| x_* - x_{\ttrue} \| \leq C \| (I - b_1[x_0]a_1[x_0]) (x_0 -
 x_{\ttrue})\|.
 \label{eq:errest}
\end{equation}

Unfortunately this is an expensive method to implement as the
computational cost of each term $b_n[x_0]$ in the inverse Born series (see
\eqref{eq:iborn:coeff}) increases exponentially with $n$. Indeed if
applying the forward Born operator $a_n[x_0]$ requires $n$ forward problem
solves (as is the case for the Schr\"odinger problem), an application of
the inverse Born operator $b_n[x_0]$ involves $2^{n-1}-1$ forward problem
solves. 

\subsection{Restarted inverse Born series (RIBS)}
A natural idea to reduce the cost of inverse Born series is to use
the $k-$th iterate of the inverse Born series \eqref{eq:it:ibs} as the
starting guess for a fresh run of inverse Born series. This gives rise
to the following class of iterative methods:
\begin{equation}
 \left\{
 \begin{aligned}
 x_0 &= \text{given},\\
 x_{n+1} &= x_n + \sum_{j=1}^k b_j[x_n] (y_{\tmeas} - f(x_n))^{\otimes j},
 ~\text{for $n\geq 0$},
 \end{aligned}
 \right.
 \label{eq:ribs}
\end{equation}
which we denote by RIBS($k$).

If $f$ is a differentiable mapping and we choose $b_1[x_n] =
(f'(x_n))^\dagger$ (where the sign $\dagger$ stands for a regularized
pseudoinverse of $f'(x_n)$),
the RIBS(1) method is in fact the Gauss-Newton method:
\begin{equation}
 \left\{
 \begin{aligned}
 x_0 &= \text{given},\\
 x_{n+1} &= x_n + f'(x_n)^\dagger (y_{\tmeas} - f(x_n)),
 ~\text{for $n\geq 0$},
 \end{aligned}
 \right.
 \label{eq:ribs1}
\end{equation}
and is quadratically convergent in a neighborhood of $x_{\ttrue}$ under
fairly mild conditions on $f$ (for $X$ and $Y$ finite dimensional, see
e.g.  \cite{Deuflhard:2011:NMNP}).

If in addition to choosing $b_1[x_n] = (f'(x_n))^\dagger$ we have
$a_2[x_n] = f''(x_n)/2$, the RIBS(2) method can be written as
\begin{equation}
 \left\{
 \begin{aligned}
 x_0 &= \text{given},\\
 x_{n+1} &= x_n - f'(x_n)^\dagger \left[r_n - \frac{1}{2}f''(x_n) ( f'(x_n)^\dagger
 r_n, f'(x_n)^\dagger r_n) \right] ,
 ~\text{for $n\geq 0$},
 \end{aligned}
 \right.
 \label{eq:ribs2}
\end{equation}
where $r_n \equiv y_\tmeas - f(x_n)$.  This is the so called Chebyshev-Halley method,
which has been studied before by \citet*{Hettlich:2000:SMIP} in the
context of inverse problems. This method is guaranteed to converge
cubically when $f''$ is Lipschitz continuous \cite{Hettlich:2000:SMIP}.

\begin{remark} 
Although the inverse Born series, and the Gauss-Newton and
Chebyshev-Halley methods are guaranteed to converge (under appropriate
assumptions), the limits may be different. The only case where we know that
these methods converge to the same $x_* = x_{\ttrue}$ is when $X=Y$,
the mapping $f$ is invertible in a neighborhood of $x_{\ttrue}$ and the
initial iterate $x_0$ is sufficiently close to $x_{\ttrue}$.
\end{remark}

\subsection{Numerical experiments on a Neumann series toy problem}
Here we compare the performance of inverse Born series, Gauss-Newton and
Chebyshev-Halley on the Neumann series problem discussed in
section~\ref{sec:neumann}. We used for discrete Laplacian $\bL$ the matrix
\[
 \bL = \begin{bmatrix} 
   -3 & 1\\
   1 &-3 & 1\\
   &\cdots & \\
    & 1 & -3 & 1\\
    & & 1 & -3 \end{bmatrix}
   \in \real^{256 \times 256}.
\]
The true parameter is a vector with zero mean, independent, normal distributed entries
and standard deviation $0.1$. The measurement operator $\bM$ is a $256
\times 8$ matrix with zero mean, independent, normal distributed entries
and standard deviation $1$. For the inverse Born series, $b_1$ is a
pseudoinverse of the Jacobian of the forward problem, where the singular
values smaller than $10^{-6}$ times the largest singular value (of the
Jacobian) are treated as zeroes. The same pseudoinverse is applied to the Jacobian
matrices involved in the Gauss-Newton and Chebyshev-Halley methods. The initial
guess for all the methods is $\bx_0 = \bzero$. For each method we
display in figure~\ref{fig:neumann} (a) the quantity $\|\bx_n
- \bx_*\|$. Since we do not have access to the limiting iterate, we
simply took one more step of each method and used it instead of $\bx_*$.
The residual terms $\|f(\bx_n)-f(\bx_{\ttrue})\|$ are shown in
figure~\ref{fig:neumann} (b). As expected, we see linear
convergence for the iterates and the residuals from the truncated
inverse Born series method. Also the first Gauss-Newton (resp.
Chebyshev-Halley) iterate error and residual matches that of the first
(resp. second) inverse Born series iterate. The Gauss-Newton method has
the expected quadratic convergence of the error, while the
Chebyshev-Halley exhibits super-quadratic convergence of the error.

\begin{figure}[!hbpt]
\centering
\begin{tabular}{cc}
\includegraphics[width=0.49\textwidth]{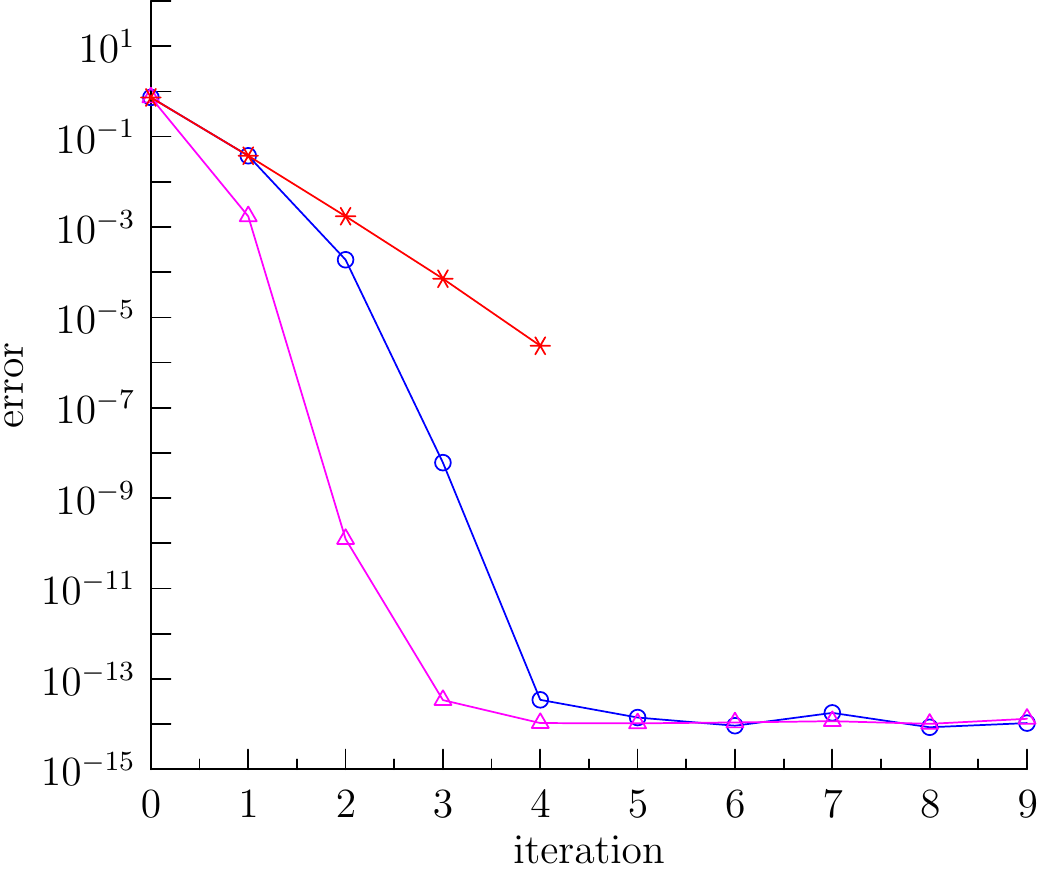} &
\includegraphics[width=0.49\textwidth]{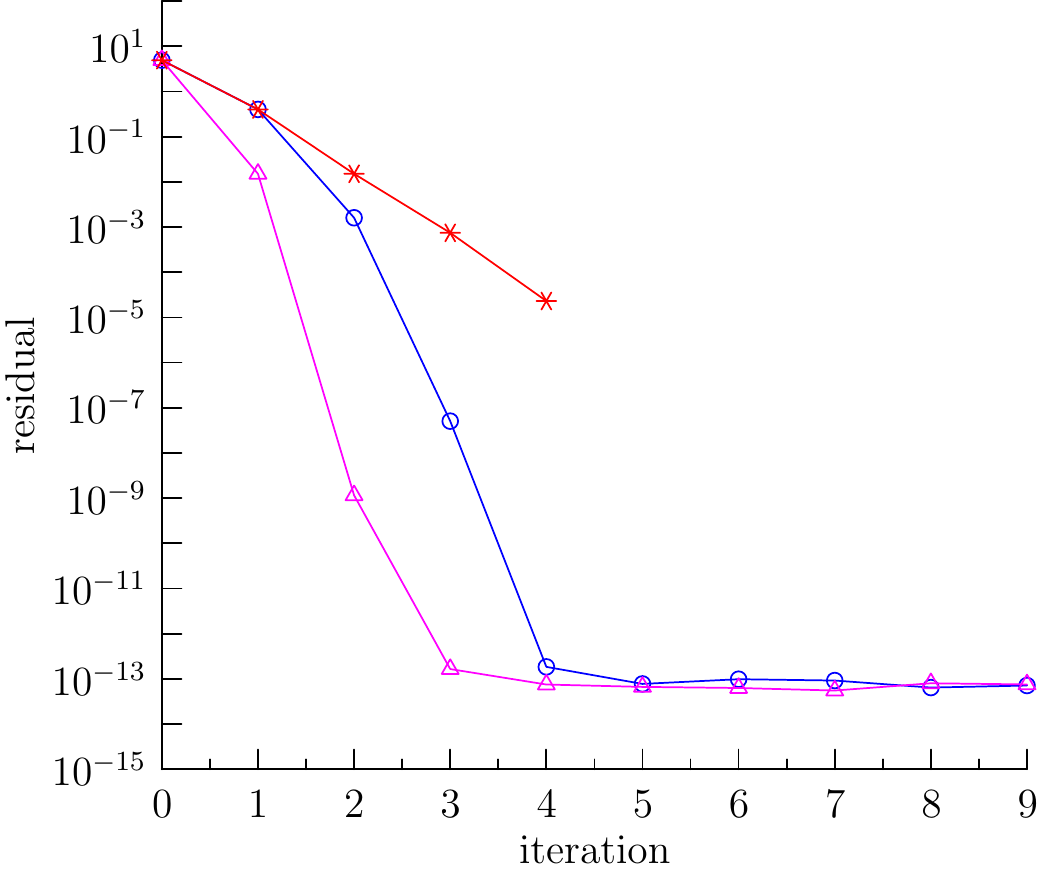}\\
(a) & (b)
\end{tabular}
\caption{Convergence of (a) iterates $\|\bx_n-\bx_*\|$ and (b) residuals
$\|f(\bx_n)-f(\bx_{\ttrue})\|$, for the inverse Born series
\textcolor{red}{($*$)}, Gauss-Newton \textcolor{blue}{($\circ$)} and
Chebyshev-Halley \textcolor{magenta}{($\triangle$)} methods. These methods are
applied to the Neumann series problem of section~\ref{sec:neumann}.}
\label{fig:neumann}
\end{figure}

\section{Forward and inverse Born series for the Schr\"odinger problem\\
with discrete internal measurements} \label{sec:sparsebounds} 

Recall from section~\ref{sec:bornconv} that local convergence of the
forward and inverse Born series follows from showing that the forward
Born operators $a_n$ satisfy bounds of the type \eqref{eq:fwdest}. We
show in section~\ref{sec:bounds:conv} that bounds of the type
\eqref{eq:fwdest} hold for the operators $a_n$ for the Schr\"odinger
problem with discrete internal measurements (defined in
\eqref{eq:bornsparse}). Then we report in section~\ref{sec:bounds:num} a
numerical approximation to the convergence radius of inverse Born
series, in a setup related to the hydraulic tomography application of
section~\ref{sec:ht}.

\subsection{Bounds on the forward Born operators}\label{sec:bounds:conv}
We recall from section~\ref{sec:qsparse} that the parameter space for
this problem is $X=L^\infty(\Omegat)$ where $\Omegat \subset \Omega$ and
the distance between $\partial \Omega$ and $\partial \Omegat$ is
positive. The difference between the unknown and the reference
Schr\"odinger potentials is assumed to be supported in $\Omegat$. The
measurements space is $Y = \complex^{N \times N}$ where $N$ is the
number of sources used and the norm is the entry-wise $\ell_\infty$ norm
of a matrix in $\complex^{N \times N}$.

The proof of lemma~\ref{lem:aninfbound} below follows a pattern similar
to \cite{Moskow:2008:CSI}. There are two main differences. The first is
that we work with finitely many measurements. The second is that we
allow the (possibly complex) reference Schr\"odinger potential $q_0$ to
be in $L^\infty(\Omega)$, whereas in \cite{Moskow:2008:CSI} the
reference potential is assumed to be constant and real. The bound
\eqref{eq:fwdest} immediately gives a smallness condition that is
sufficient for convergence of the forward Born series. The smallness
condition we obtain is identical to that in \cite{Moskow:2008:CSI}. This
is to be expected because the underlying equation is the same and only
the measurements differ.

To prove lemma~\ref{lem:aninfbound}, we need that the reference
Schr\"odinger potential $q_0(\bx) \in L^\infty(\Omega)$ is such that the
only solution to
  \begin{equation}
   \left\{
    \begin{aligned}
     -\Delta u + q_0 u &= 0, &&\text{in $\Omega$},\\
     u &=0, &&\text{on $\partial \Omega$,}
    \end{aligned}
   \right.
   \label{eq:qhomog}
  \end{equation}
is $u=0$. Such $q_0$ are sometimes called ``non-resonant'' and we assume
that all the Schr\"odinger potentials that we deal with in what follows
are non-resonant. We also need two properties for the Green function
$G_{q_0}(\bx,\by)$ for the Schr\"odinger equation (as
defined in section~\ref{sec:qsparse}):
\begin{enumerate}
 \item The function $\bx \mapsto G_{q_0}(\bx,\by)$ is in $L^1(\Omega)$
 for all $\by \in \Omega$.
 \item The function $\by \mapsto \| G_{q_0} (\cdot,\by)
 \|_{L^1(\Omega)}$ is in $L^\infty(\Omega)$.
\end{enumerate}
These properties can be easily verified in both $\real^2$ and $\real^3$ for $G_0$ (i.e. when $q_0 \equiv
0$) and hold for general bounded $q_0$. Indeed,  we have $(\Delta + q_0)
(G_{q_0} - G_0) = - q_0 G_0$. Since the right hand side belongs to
$L^2(\Omega)$, the difference $G_{q_0} - G_0$ must be in
$H^2_{\text{loc}}(\Omega)$ by standard elliptic regularity estimates
(see e.g.  \cite{Evans:2010:PDE}) and therefore continuous (by Sobolev
embeddings).  This argument shows that $(G_{q_0} - G_0)(\bx,\by)$ is
continuous as function of $\bx$ and for all $\by$. By reciprocity
$G_{q_0} - G_0$ is continuous on $\Omega \times \Omega$. Therefore
$G_{q_0}$ satisfies the desired properties.

We can now show boundedness of the operators $a_n$ for the Schr\"odinger
equation with discrete measurements. The proof of the following lemma is
similar to that in \cite{Moskow:2008:CSI}.

\begin{lemma}\label{lem:aninfbound}
Let $q_0(\bx)$ be a (possibly complex) non-resonant Schr\"odinger
potential. Then the operators $a_n$ defined in \eqref{eq:bornsparse}
satisfy the bounds
\begin{equation}\label{eq:aninfbound}
||a_n|| \leq \alpha\mu^{n},
\end{equation}
with $\alpha=\nu/\mu$, and where $\nu$ and $\mu$ are constants depending
on $\Omega$ and $q_0$ only (see equations \eqref{eq:nuinf} and
\eqref{eq:muinf} below for their definition). The norm on $a_n$ is the
operator norm in $\cL( X^{\otimes n}, Y)$, with parameter space $X$ and
data space $Y$ as in section~\ref{sec:qsparse}.
\end{lemma}

\begin{proof}
Consider $\eta_1\otimes\cdots\otimes \eta_n\in (L^\infty(\Omegat))^{\otimes n}$ and observe
\begin{equation}\label{eq:aninf}
\begin{aligned}
||a_n(\eta_1\otimes\cdots\otimes \eta_n)|| 	&=
\sup_{i,j}|(a_n(\eta_1\otimes\cdots\otimes \eta_n))_{i,j}| \\
			&\leq ||\eta_1\otimes\cdots\otimes
			\eta_n||_{(L^\infty(\Omegat))^{\otimes n}}
			\sup_{i,j} \int_{{\Omegat}^{n+2}}\big|G_{q_0}(\bx,\by_1)\cdot\\
			&\qquad\cdots G_{q_0}(\by_{n-1},\by_n)
			G_{q_0}(\by_n,\bz) \phi_{i}(\bz)
			\phi_{j}(\bx)\big|d\bz d\by_1 \cdots d\by_n
			d\bx.
\end{aligned}
\end{equation}
We start by estimating $||a_1||$ as follows,
\begin{equation*}\label{eq:a1estimate}
\begin{aligned}
||a_1||  & \leq \sup_{i,j} \int_{{\Omegat} \times {\Omegat} \times
{\Omegat}}{\big|G_{q_0}(\bx,\by_1)G_{q_0}(\by_1,\bz)\phi_{i}(\bz)\phi_{j}(\bx)
\big|d\bz d\by_1 d\bx} \\
		& \leq
		\sup_{i,j}\int_{\Omegat}\int_{\Omegat}\big|G_{q_0}(\by_1,\bz)\phi_{i}(\bz)\big|d\bz
		\int_{\Omegat}\big|G_{q_0}(\bx,\by_1)\phi_{j}(\bx)\big|d\bx d\by_1 \\
		& \leq \sup_i
		\Bigg(\sup_{\bx\in{\Omegat}}{\int_{\Omegat}\big|G_{q_0}(\bx,\by)\phi_i(\by)\big|
		d\by}\Bigg)^2|\Omega|.
\end{aligned}
\end{equation*}
Since $q_0$ is assumed to be non-resonant and using that $\phi_i \in
L^\infty(\Omega)$, the quantity
\begin{equation}\label{eq:nuinf}
\begin{aligned}
\nu 	&=
\Bigg(\sup_i\sup_{\bx\in{\Omegat}}{\int_{\Omegat}\big|G_{q_0}(\bx,\by)\phi_i(\by)\big|d\by}\Bigg)^2|\Omega|
\end{aligned}
\end{equation}
is bounded. We have established that $||a_1|| \leq \nu$.

For the remaining Born operators, we proceed recursively. Considering again
\eqref{eq:aninf} for $n \geq 2$, we have
\begin{equation*}
\begin{aligned}
||a_n||	 & \leq \sup_{i,j} \int_{\Omegat^{n+2}}\big|G_{q_0}(\bx,\by_1)G_{q_0}(\by_1,\by_2)\cdot\\
		 & \qquad\cdots G_{q_0}(\by_{n-1},\by_n)G_{q_0}(\by_n,\bz)
		 \phi_{i}(\bz) \phi_{j}(\bx) \big|d\bz d\by_1 \cdots d\by_n d\bx \\
		 & \leq \sup_{i,j} \Bigg(\sup_{\by_1 \in{\Omegat}}\int_{\Omegat}{\big|G_{q_0}(\bx,\by_1)\phi_{j}(\bx)\big|d\bx}\Bigg) \Bigg(\sup_{\by_n \in
		 {\Omegat}}\int_{\Omegat}{\big|G_{q_0}(\by_n,\bz)\phi_{i}(\bz)\big|d\bz}\Bigg) \\
		 & \cdot \int_{{\Omegat}^n}{\big|G_{q_0}(\by_1,\by_2)\cdots G_{q_0}(\by_{n-1},\by_n)\big| d\by_1 \cdots d\by_n} \\
		 & \leq \Bigg( \sup_i \sup_{\bx \in {\Omegat}}
		 \int_{\Omegat} \big|G_{q_0}(\bx,\by)\phi_i(\by)\big|d\by\Bigg)^2 I_{n-1}
\end{aligned}
\end{equation*}
where
\begin{equation*}
I_{n-1} = \int_{{\Omegat}^n}{\big|G_{q_0}(\by_1,\by_2)\cdots
G_{q_0}(\by_{n-1},\by_n)\big| d\by_1 \cdots d\by_n}.
\end{equation*}
Estimating $I_{n-1}$ we find that
\begin{equation*}
\begin{aligned}
I_{n-1} 	&\leq \sup_{\by_{n-1} \in {\Omegat}} \int_{\Omegat}{
\big|G_{q_0}(\by_{n-1},\by_n)\big|d\by_n} \cdot
\int_{{\Omegat}^{n-1}}{\big|G_{q_0}(\by_1,\by_2)\cdots
G_{q_0}(\by_{n-2},\by_{n-1})\big| d\by_1 \cdots d\by_{n-1}} \\
		&\leq \mu I_{n-2},
\end{aligned}
\end{equation*}
where the quantity
\begin{equation}\label{eq:muinf}
\mu = \sup_{\bx \in {\Omegat}} ||G_{q_0}(\bx,\cdot)||_{L^1({\Omegat})}
\end{equation}
is finite by the properties that $G_{q_0}$ satisfies.  Finally, noting that
\begin{align*}
I_1 &= \int_{{\Omegat} \times {\Omegat}}{\big|G_{q_0}(\by_1,\by_2)\big|d\by_1 d\by_2}\nonumber \\
		&\leq \mu |\Omega|,
\end{align*}
it follows that
\begin{equation*}
I_{n-1} \leq |\Omega|\mu^{n-1},
\end{equation*}
and thus \begin{equation*}
||a_n|| \leq \left( \sup_i \sup_{\bx \in {\Omegat}}
||G_{q_0}(\bx,\cdot)||_{L^1(B_\rho(\bx_i))}\right)^2|\Omega|\mu^{n-1}=\alpha\mu^{n}.
\end{equation*}
\end{proof}

\begin{remark}[$L^p$ Bounds] 
Bounds similar to those in lemma~\ref{lem:aninfbound} can be proven
when the parameter space is $X=L^2(\Omega)$ and the data space is
$Y=\complex^{N\times N}$, endowed with the Frobenius norm. Once we
have bounds for the $\infty$ and $2$ norms, it is possible to invoke the
Riesz-Thorin theorem (as in \cite{Moskow:2008:CSI}) to show bounds for
$2 \leq p \leq \infty$ by interpolation. In this case the data space is
$X=L^p(\Omega)$ and the parameter space is $Y = \complex^{N \times N}$,
endowed with the entry-wise $p-$norm (i.e. the $p-$norm of the
$\complex^{N^2}$ vector obtained by stacking the columns of a matrix in
$\complex^{N \times N}$).
\end{remark} 

Having established norm bounds on the operators $a_n$ for the discrete
measurements Schr\"odinger problem, we can apply the results from
section~\ref{sec:bornconv} to establish local convergence of the forward
Born series, local convergence of the inverse Born series (provided the
linear operator $b_1$ used to prime the series has sufficiently small
norm, see theorem~\ref{thm:smallness:inv}), stability of the inverse
Born series (theorem~\ref{thm:stability}) and even an error estimate
(theorem~\ref{thm:error}). The actual choice of $b_1$ is discussed in
section~\ref{sec:numerics}.

\subsection{Numerical illustration}\label{sec:bounds:num}
Applying theorem~\ref{thm:smallness:inv} to the Schr\"odinger problem
with discrete measurements, we can expect the inverse Born series to
converge when the difference $d$ between the data for the unknown and
reference Schr\"odinger potentials satisfies
\[
 \| d \| \leq \frac{1}{(1+\alpha)\mu \| b_1\|},
\]
where the constants $\alpha=\nu/\mu$ and $\mu$ are constants defined by
\eqref{eq:nuinf} and \eqref{eq:muinf} and the norms are as in
section~\ref{sec:qsparse}. 

In preparation for the application to hydraulic tomography, we consider
the setup depicted in figure~\ref{fig:setup} with computational domain
$\Omega = [0,1]^2$. The distance between
$\Omega$ and $\Omegat$ is $\epsilon \in [0,1/4]$ and the sources
$\phi_i$ are supported in disks of radius $0.05$ with centers
$(0.2k,0.2l)$, for $k,l=1,\ldots,4$. The sources are $\phi_i(\bx) =
\phi(\bx-\bx_i)$ where $\bx_i$ is the center of the disk support and
$\phi$ is an infinitely smooth function with $0 \leq \phi(\bx) \leq 1$.  Although
theorem~\ref{thm:smallness:inv} allows for the supports of the sources
to overlap, we take them to be disjoint as this is the case in
the hydraulic tomography application. 

The constants $\mu$ and $\nu$ are approximated by solving appropriate
(forward) Schr\"odinger problems with $q_0 = 0$.  The grid we use for
this purpose is uniform and consists of the nodes $(kh,lh)$ for
$k,l=0,\ldots,400$ and $h=1/400$. We display in
figure~\ref{fig:radiusconv} the radius of convergence of the inverse
Born series predicted by theorem~\ref{thm:smallness:inv}, assuming $\|
b_1 \| = 1$. We observe that the radius of convergence increases as
$\epsilon$ increases, or in other words, the larger the region where we
assume the Schr\"odinger potential is known, the larger the
perturbations in the data the method can handle.

\begin{figure}
 \centering 
  \includegraphics[width=0.3\textwidth]{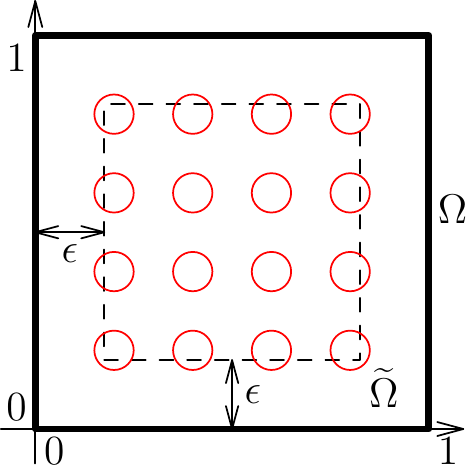}
 \caption{Setup for the numerical experiments with the Schr\"odinger
 problem with internal measurements. The domain $\Omega$ is the unit
 square. The domain $\Omegat$ where the Schr\"odinger potential is
 unknown is in dotted line and its boundary $\partial\Omegat$
 is at a distance $\epsilon$ from $\partial\Omega$. The supports of the
 functions used as source terms/measurements are the red circle.}
  \label{fig:setup}
\end{figure}

\begin{figure}
\centering
\includegraphics[width=0.65\textwidth]{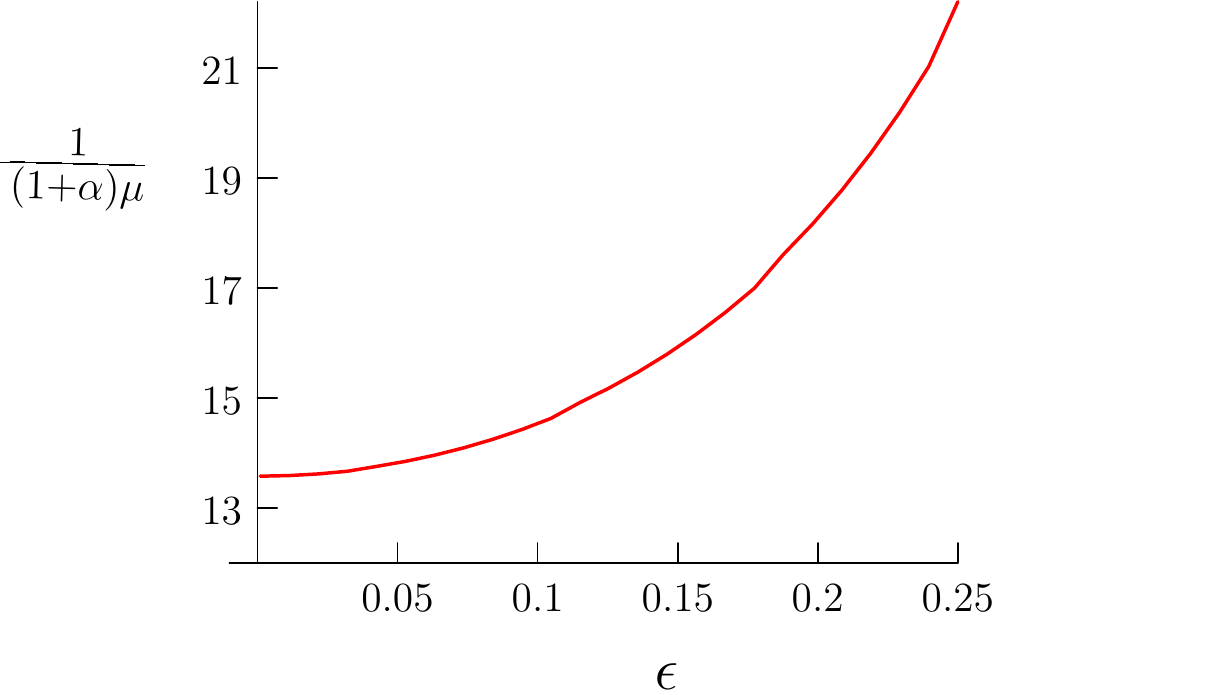}
\caption{Numerical approximation of the radius of convergence
for the inverse Born series for the Schr\"odinger problem with discrete
internal measurements and assuming $\|b_1\| \geq 1$. The reference
Schr\"odinger potential is $q_0=0$ and the setup is that given in
figure~\ref{fig:setup}.}\label{fig:radiusconv}
\end{figure}

\section{Application to transient hydraulic tomography}\label{sec:ht}

Consider an underground aquifer confined in a bounded domain $\Omega$.
The head or hydraulic pressure $u_i(\bx,t)$ in the aquifer due to
injecting water in the $i-$th well satisfies the equation
\begin{equation}\label{eq:hydtom}
\left\{\begin{aligned}
S\frac{\partial u_i}{\partial t} &= \nabla\cdot(\sigma\nabla u_i) -
\phi_i, &&\text{ for } \bx\in\Omega,~t>0,\\
 u_i(\bx,t) 	 &= 0,  &&\text{ for } \bx\in\partial\Omega,~t>0,\\
 u_i(\bx,0) 	 &= g(\bx),  &&\text{ for } \bx\in\Omega.
 \end{aligned}\right.
\end{equation}
where $i=1,\ldots,N$. Here we assume there are no sources or leaks of
water in the aquifer, other than those prescribed at the wells.  Hence
the source term $\phi_i(\bx,t)$ is supported at the $i-$th well and
represents the water injected at the $i-$th well.  The physical
properties of the aquifer are modeled by the storage coefficient
$S(\bx)$ and the hydraulic conductivity $\sigma(\bx)$. The initial head
(at $t=0$) is given by $g(\bx)$.

The inverse problem of hydraulic tomography that we consider here, is to
determine the coefficients $\sigma$ and $S$ from knowledge of the
discrete internal measurements
\begin{equation}
M_{i,j}(t)=\int_\Omega \phi_j(\bx,t) * u_i(\bx,t)d\bx, ~i,j=1,\ldots,N,
\label{eq:tdmeas}
\end{equation}
where the convolution is in time. Physically these measurements
correspond to time domain measurements at the $j-$th well of a spatial
average of the hydraulic pressure $u_i$ generated by injecting in the
$i-$th well. Here for simplicity, we use for the impulse response (in
time) of the $j-$th measurement well the function $\phi_j(\bx,t)$. In a
more general setup, the injection and measurement ``well functions'' can
be different.

\subsection{Reformulation as a discrete internal measurements Schr\"odinger problem}
\label{sec:ht_schro}
The frequency domain version of problem \eqref{eq:hydtom} is
\begin{equation}\label{eq:hydtom-freq}
\left\{
\begin{aligned}
\nabla \cdot (\sigma \nabla \hu_i) - \ii\omega S \hu_i &= \hphi_i,
&&\text{ for } \bx\in\Omega,\\
\hu_i &= 0, &&\text{ for } \bx\in\partial\Omega,
\end{aligned}\right.
\end{equation}
where the hat denotes Fourier transform in time, i.e.
\[
\hu_i(\bx,\omega) = \int_\real u_i(\bx,t)e^{-\ii\omega t}dt
\quad\text{ and }\quad
\hphi_i(\bx,\omega) = \int_\real \phi_i(\bx,t)e^{-\ii\omega t}dt.
\]
The inverse problem is now to recover $\sigma$ and $S$ from the
discrete internal measurements
\begin{equation}
\hM_{i,j}(\omega) = \int_\Omega
\hphi_j(\bx,\omega)\hu_i(\bx,\omega)d\bx,
\label{eq:fdmeas}
\end{equation}
which is the Fourier transform in time of the discrete internal
measurements for the time domain problem \eqref{eq:tdmeas}.

Next we use the Liouville transformation by defining $v_i =
\sigma^{1/2}\hu_i$.  If $\hu_i$ satisfies
\eqref{eq:hydtom-freq} then $v_i$ must satisfy the Schr\"odinger
equation
\begin{equation}\label{eq:hydtom-liouv}
\left\{
\begin{aligned}
\Delta v_i -\left(\frac{\Delta
\sigma^{1/2}}{\sigma^{1/2}}+\frac{\ii\omega S}{\sigma}\right)v_i &=
\frac{\hphi_i}{\sigma^{1/2}}, &&\text{
for } \bx\in\Omega,\\
v_i &= 0, &&\text{ for } \bx\in\partial\Omega.
\end{aligned}
\right.
\end{equation}
The internal measurements $\hM_{i,j}(\omega)$ can now be
expressed in terms of $v_i$ as
\[
\hM_{i,j}(\omega) = \int_\Omega \hphi_j(\bx,\omega)\hu_i(\bx,\omega)d\bx
= \int_\Omega
\frac{\hphi_j(\bx,\omega)}{\sigma^{1/2}(\bx)}v_i(\bx,\omega)d\bx.
\]
Hence the measurements $\hM_{i,j}(\omega)$ are of the form defined in
\eqref{eq:sparsedata} with test functions $\hphi_i/\sigma^{1/2}$
(modeling both injection and measurement).

If we do have access to the inside of the wells (i.e. $\supp \hphi_i$),
it is reasonable to assume that $\sigma$ is known in $\supp \hphi_i$.
Hence the test functions $\hphi_i/\sigma^{1/2}$ are known and we can use
any method for solving the inverse Schr\"odinger problem with discrete
data to obtain an approximation to the complex Schr\"odinger potential 
\begin{equation}
Q(\bx;\omega) = 
\frac{\Delta
\sigma^{1/2}}{\sigma^{1/2}}+\frac{\ii\omega S}{\sigma}, 
~\text{for $\bx \in \Omega$}.
\label{eq:cpot}
\end{equation}

\begin{remark}
A limitation of transforming the hydraulic tomography problem into an
inverse Schr\"odinger problem is that the conductivity $\sigma$ appears
as $\Delta \sigma^{1/2}/\sigma^{1/2}$ in the Schr\"odinger potential.
Therefore any high (spatial) frequency components in $\sigma^{1/2}$ are
magnified. The resulting Schr\"odinger potential can easily fall outside of the
radius of convergence of the inverse Born series. It may be possible to
overcome this limitation if we apply the inverse Born series to the
hydraulic tomography problem directly (i.e. without doing the Liouville
transform).
\label{rem:sigma_scale}
\end{remark}

\subsection{Recovery of $S$ and $\sigma$ from one frequency}
\label{sec:onefreq}

Once we have approximated $Q(\bx;\omega)$ for a single (known) frequency $\omega$,
the real part of $Q(\bx;\omega)$ can be used to estimate the hydraulic conductivity
$\sigma$. This can be achieved by solving for $\sigma^{1/2}(\bx)$ in the equation
\[
\Delta \sigma^{1/2} - \Re(Q(\bx;\omega))\sigma^{1/2} = 0,
\]
on the aquifer without the wells, i.e.
\[
 \Omega' \equiv \Omega \backslash \bigcup_{i=1}^n \supp \hphi_i,
\]
and with Dirichlet boundary conditions at $\partial \Omega'$ determined
from the (assumed) knowledge of $\sigma$ at the measurement wells and at
$\partial \Omega$. An estimate of the storage coefficient $S$ from
$\Im(Q(\bx;\omega))$ and $\sigma(\bx)$ follows since \[S(\bx) =
\sigma(\bx) \Im(Q(\bx;\omega))/\omega.\]

In principle, measurements $\hM_{i,j}(\omega)$ for one single frequency
are enough to find both parameters $\sigma(\bx)$ and $S(\bx)$.
Unfortunately, this procedure seems to be much more sensitive to changes
in $\sigma$ than to changes in $S$. This is due to $\Delta \sigma^{1/2}$
appearing in the expression of $Q(\bx;\omega)$ (see
remark~\ref{rem:sigma_scale}). We deal with this problem by using data
for two frequencies as is explained below.

\subsection{Recovery of $S$ and $\sigma$ from two frequencies}
\label{sec:twofreq}
Here the data we have is $\hM_{i,j}(\omega_1)$ and $\hM_{i,j}(\omega_2)$
for two frequencies $\omega_1\neq \omega_2$ and we use it to solve two
discrete measurements Schr\"odinger problems for $Q(\bx;\omega_1)$ and
$Q(\bx;\omega_2)$, for $\bx \in \Omega$. A good rule of thumb is to
choose the frequencies so that $\omega_1$ is sufficiently low to make
$\Re(Q(\bx;\omega_1))$ the largest term in $Q(\bx;\omega_1)$ and
$\omega_2$ is sufficiently large to make $\Im(Q(\bx;\omega_2))$ the
largest term in $Q(\bx;\omega_2)$. For each point $\bx$ in $\Omega'$
(the domain without the wells), we solve for $r_1(\bx)$ and $r_2(\bx)$
in the $2\times 2$ system: 
\begin{equation}
 \begin{bmatrix}
  1 & \ii \omega_1\\
  1 & \ii \omega_2
 \end{bmatrix}
 \begin{bmatrix}
  r_1(\bx)\\
  r_2(\bx)
 \end{bmatrix}
 = \begin{bmatrix}
  Q(\bx;\omega_1)\\
  Q(\bx;\omega_2)
   \end{bmatrix}.
\end{equation}
Then to estimate the conductivity we solve for $\sigma^{1/2}$ in the
equation:
\begin{equation}
 \Delta \sigma^{1/2} - r_1(\bx) \sigma^{1/2} = 0,~\text{for $\bx \in
 \Omega'$},
 \label{eq:twofreq:r1}
\end{equation}
with Dirichlet boundary condition given by the knowledge of $\sigma$ on
$\partial \Omega'$.
Once we know $\sigma$, the storage coefficient $S$ can be easily
obtained from $r_2$, indeed: 
\begin{equation}
 S(\bx) = \sigma(\bx) r_2(\bx).
 \label{eq:Sest}
\end{equation}

\section{Numerical Experiments}\label{sec:numerics} 

We now present numerical experiments comparing inverse Born series with the
Gauss-Newton and Chebyshev-Halley methods for both the discrete internal
measurements Schr\"odinger problem (section~\ref{sec:schro_reconst}) and an
application to transient hydraulic tomography (section~\ref{sec:ht_reconst}).

\subsection{Schr\"odinger potential reconstructions from discrete
internal measurements}\label{sec:schro_reconst}

As discussed in section~\ref{sec:qsparse}, our objective is to recover
an unknown Schr\"odinger potential $q$ from the measurements $f(q)=\bD$,
where the entries $D_{i,j}$ of the $N \times N$ matrix $\bD$ are given by
\eqref{eq:sparsedata}.

We discretize the computational domain $\Omega = [0,1]^2$ with a uniform
grid consisting of the nodes $(kh,lh)$, for $k,l=0,\ldots,400$ and
$h=1/400$. We use a total of $16$ measurement functions $\phi_j$, which
are smooth and satisfy: $\|\phi_j\|_{L^\infty(\Omega)}=1$ for
$j=1,\ldots,16$; $\phi_j$ is compactly supported on a circle of radius
$\rho=0.05$; and the centers of the wells are uniformly spaced in the domain
at the points $(0.2m,0.2n)$ for $m,n=1,\ldots,4$. The Laplacian in the
Schr\"odinger equation is discretized with the usual five point finite
differences stencil and the true Schr\"odinger potential is simply evaluated at
the grid nodes. The measurements $D_{i,j} = \langle \phi_j, u_i
\rangle_{L^2(\Omega)}$ involve integrals that are approximated by the
trapezoidal rule on the grid. Measurements $f(q_0)$ for the reference
potential $q_0$ are computed in the same grid.  The data that we use for
the reconstructions is $f(q)-f(q_0)$.

The reconstructions are performed on a different (coarser) grid
consisting of the nodes $(kh_c,lh_c)$ for $k,l=0,\ldots,80$ and
$h_c=1/80$. We compare the results obtained from a truncated inverse
Born series of order 5, and 10 iterations of the Gauss-Newton and
Chebyshev-Halley methods. These three reconstructions are applied to
$F$, a coarse grid version of the map $f$. For instance, the
reconstructions for the inverse Born series are
\[
 \sum_{n=1}^k B_n ( ( f(q) - f(q_0) )^{\otimes n} ),
\]
where the coefficients $B_n$ are the inverse Born series coefficients
for the coarse grid $F$ (rather than those for the fine grid
$f$, which would be an inverse crime). For the inverse Born series, the
operator $B_1$ is a regularized pseudoinverse of $A_1$ (i.e.  the
linearization of the coarse grid forward map $F$) where the singular
values of $A_1$ which are less than $0.01$ times the largest singular
value (of $A_1$) are treated as zero.  The same regularization is used
for the Jacobians involved in the Gauss-Newton and Chebyshev-Halley
methods. We use $q_0=0$ as the reference potential for the inverse Born
series as well as the initial guess for the iterative Gauss-Newton and
Chebyshev-Halley methods. 

Figure~\ref{fig:qs_noiseless} shows the reconstructions of a real smooth
Schr\"odinger potential $-14\leq q(\bx) \leq 4$ and a real piecewise
constant potential with $-6 \leq q(\bx) \leq 12$. In both cases, the
potential and the generated data are small enough to satisfy the
hypotheses of theorem~\ref{thm:error}.  Figure~\ref{fig:qs_noise1}
displays the reconstructions of the same potentials from noisy data. The
noisy data is obtained by first generating the true data $f(q)-f(q_0)$
as above, and then perturbing it with $1\%$ zero mean additive Gaussian
noise, i.e. with standard deviation $0.01 \max_{i,j} |(f(q)-f(q_0))
_{i,j}|$.  Similarly, figure~\ref{fig:qs_noise} displays the
reconstructions with 5\% additive Gaussian noise, i.e. with zero mean
and standard deviation $0.05\max_{i,j}|(f(q)-f(q_0))_{i,j}|$. In the
experiments with noise present, the pseudoinverses of the Jacobians have
been additionally regularized to compensate for the noise level (i.e.
only singular values above $0.02$ (resp. $0.06$) times the largest
singular value are retained for inversion for 1\% (resp. 5\%) noise).

\begin{figure}
\centering
\includegraphics[width=\textwidth]{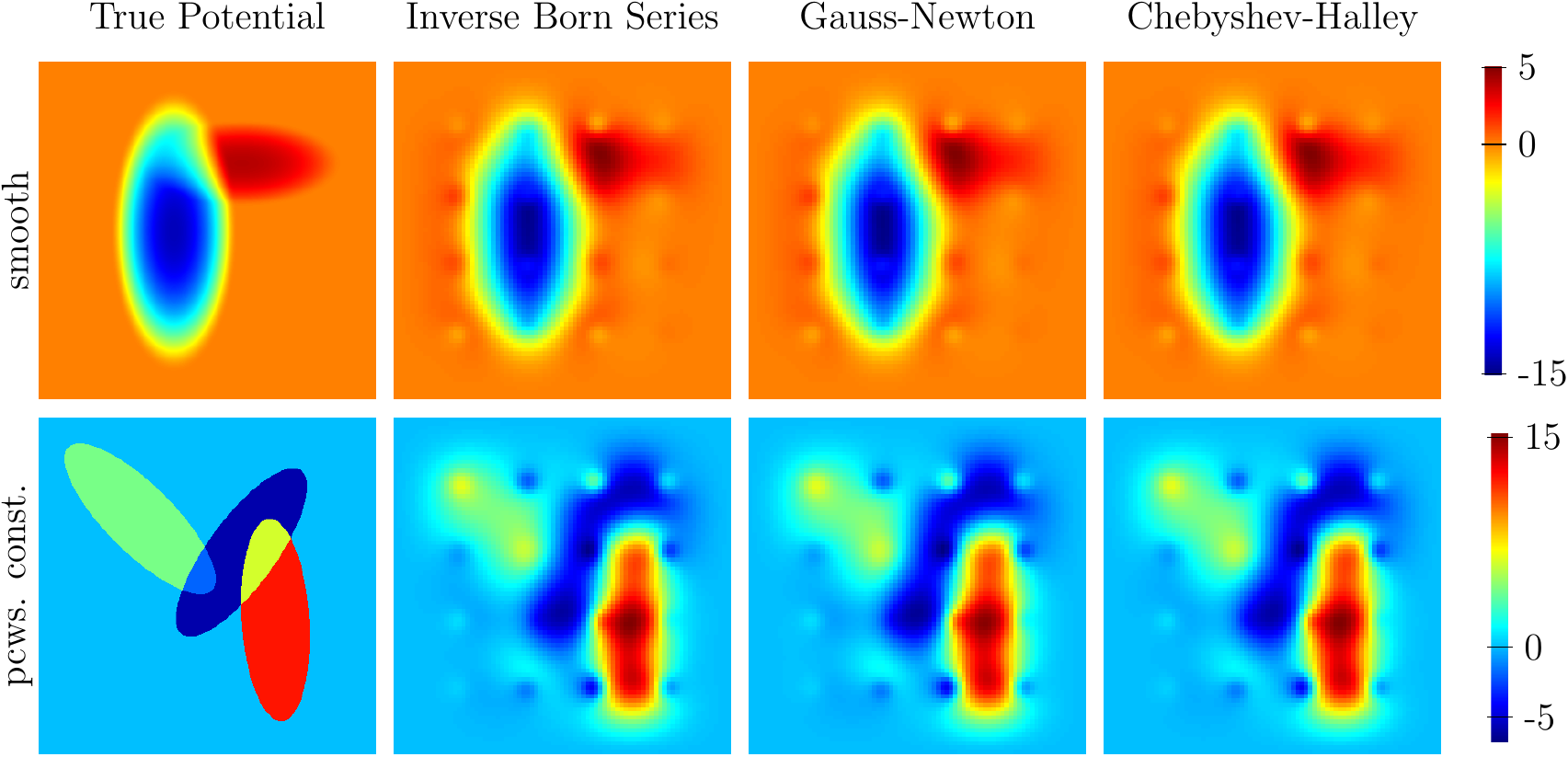}
\caption{Comparison of reconstructions of a smooth (top) and piecewise constant
(bottom) Schr\"odinger potential from discrete internal data at 16
locations and with no noise. The color scale is identical for all images in
a row.}\label{fig:qs_noiseless} 
\end{figure}

\begin{figure}
\centering
\includegraphics[width=\textwidth]{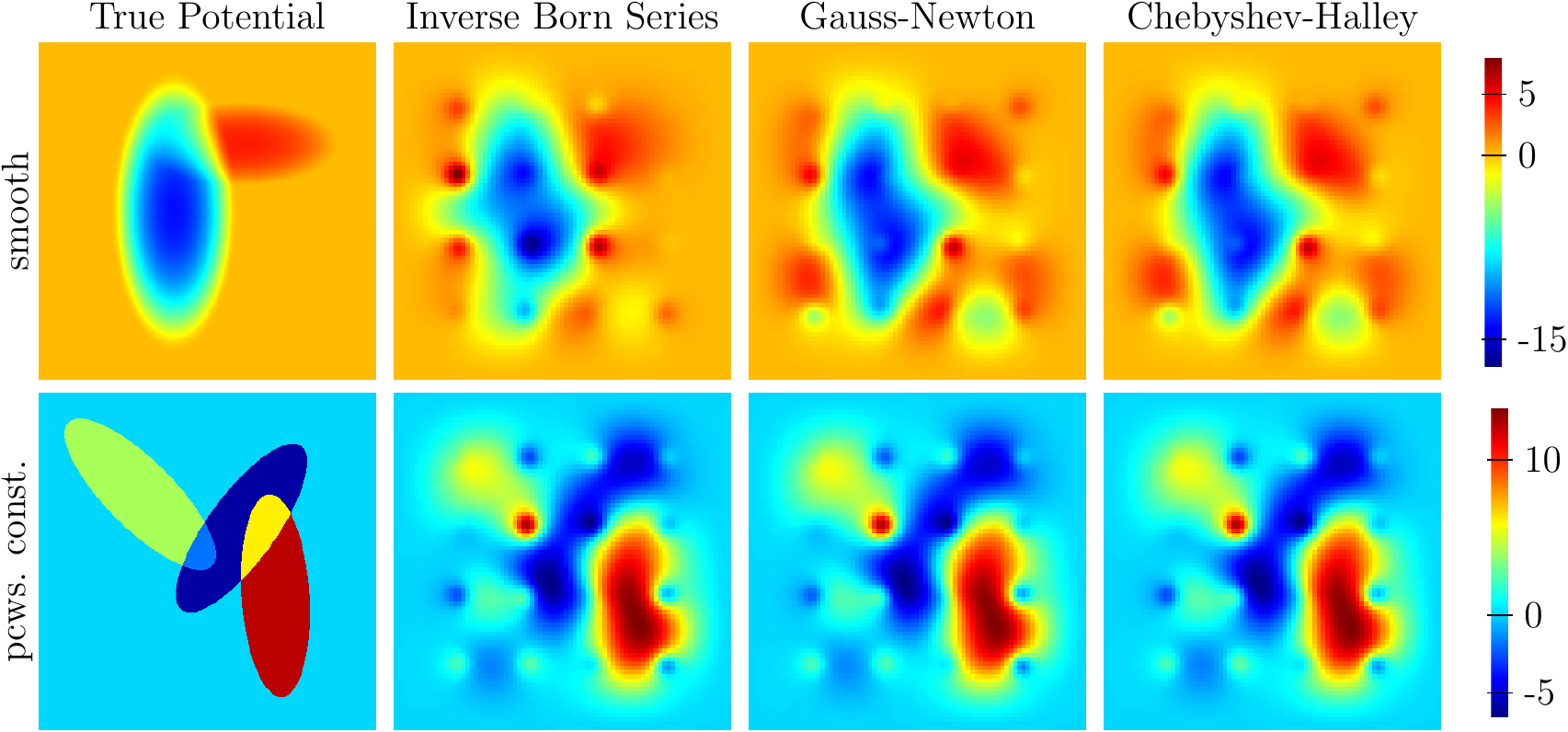}
\caption{Comparison of reconstructions of a smooth (top) and piecewise constant
(bottom) Schr\"odinger potential from discrete internal data at 16
locations and with 1\% additive Gaussian noise. The color scale is identical for all images in
a row.}\label{fig:qs_noise1}
\end{figure}

\begin{figure}
\centering
\includegraphics[width=\textwidth]{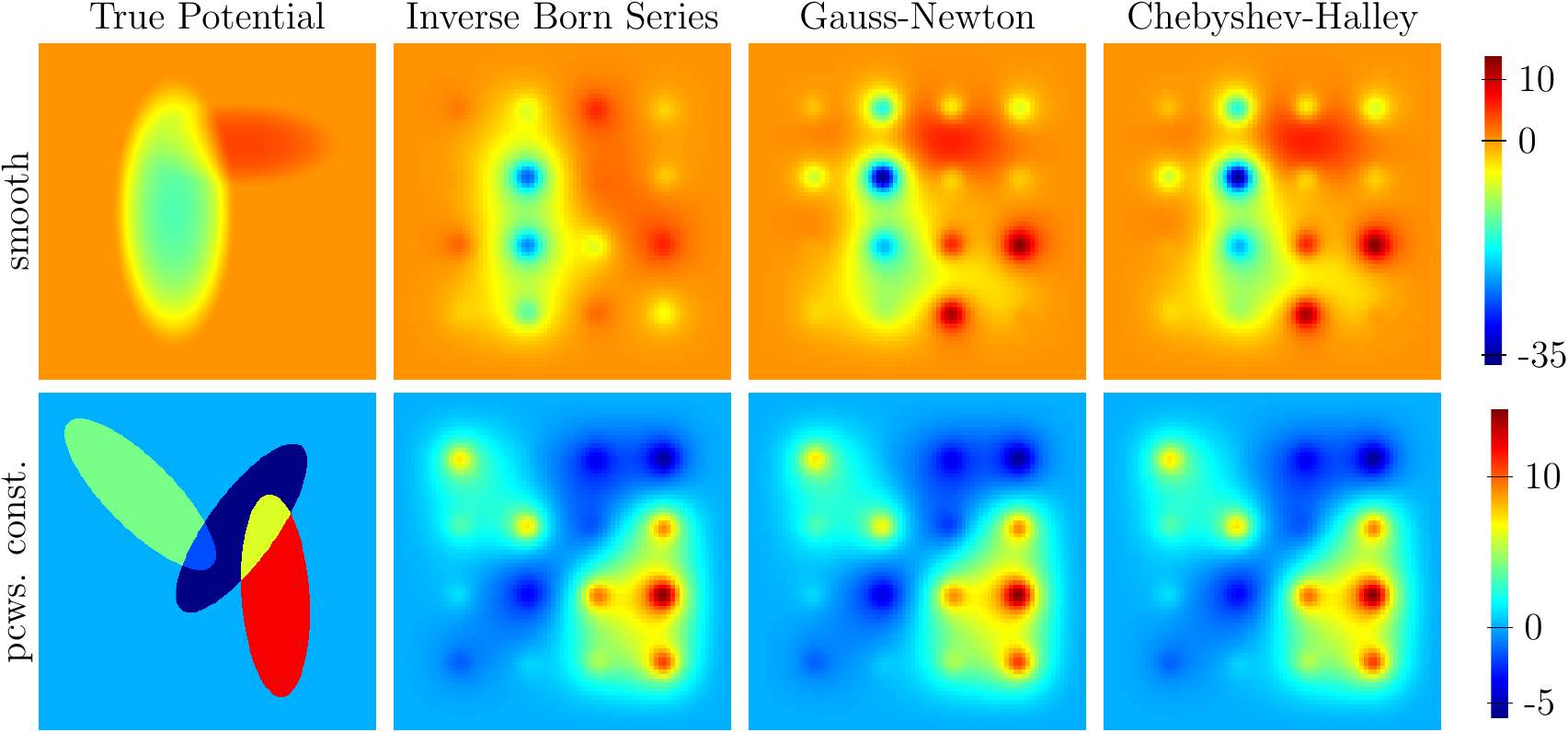}
\caption{Comparison of reconstructions of a smooth (top) and piecewise constant
(bottom) Schr\"odinger potential from discrete internal data at 16
locations and with 5\% additive Gaussian noise. The color scale is identical for all images in
a row.}\label{fig:qs_noise}
\end{figure}

\subsection{Transient hydraulic tomography}\label{sec:ht_reconst}

In the frequency domain hydraulic tomography problem (see
section~\ref{sec:ht}), the objective is to estimate the hydraulic
conductivity $\sigma(\bx)$ and the storage coefficient $S(\bx)$ from the
frequency dependent measurements $\hM_{i,j}(\omega)$ defined in
\eqref{eq:fdmeas}. 

As before, the computational domain $\Omega=[0,1]^2$ is discretized with
a uniform grid with nodes $(kh,lh)$ for $k,l=0,\ldots 400$ and $h=1/400$. The true
storage coefficient $S$ is evaluated on this grid. The discretization of
the term $\nabla \cdot [ \sigma \nabla u]$ is done through the stencil
\[
 \begin{aligned}
(\nabla \cdot [ \sigma \nabla u])(kh,lh) \approx
 &\sigma_{k+1/2,l} \frac{u_{k+1,l} - u_{k,l}}{h^2} +
 \sigma_{k-1/2,l} \frac{u_{k-1,l} - u_{k,l}}{h^2} \\
 + &\sigma_{k,l+1/2} \frac{u_{k,l+1} - u_{k,l}}{h^2} +
 \sigma_{k,l-1/2} \frac{u_{k,l-1} - u_{k,l}}{h^2},
 \end{aligned}
\]
where $u_{k,l} \approx u(kh,lh)$ and similarly for $\sigma$. This means
that the true conductivity is evaluated at the midpoints of the
horizontal and vertical edges of the grid. The boundary points have a
different stencil that takes into account the homogeneous Dirichlet
boundary conditions, and that we do not include here for the sake of
clarity.

The frequency domain measurement functions $\hphi_i(\bx,\omega)$ we use
are, for simplicity, independent of the frequency $\omega$ and are given
in $\bx$ by the same 16 compactly supported smooth functions described
in section~\ref{sec:schro_reconst}. The measurements $\hM_{i,j}(\omega)
= \langle \hphi_j , \hu_i \rangle_{L^2(\Omega)}$ involve integrals over
$\Omega$ that are evaluated by using the trapezoidal rule on the same
grid that is used for the forward simulations. Recalling
section~\ref{sec:ht_schro}, the measurements $\hM_{i,j}(\omega)$ can
also be viewed as discrete internal measurements of a Schr\"odinger
field $v_i$ (see \eqref{eq:hydtom-liouv}) associated with the potential
$Q(\bx;\omega)$ defined in \eqref{eq:cpot} i.e. ${\bf\hM(\omega)} =
f(Q(\bx;\omega))$ with well functions
$\hphi_i/\sigma^{1/2}$. We also compute measurements for the reference
potential $Q_0=0$ on this grid using the well functions
$\hphi_i/\sigma^{1/2}$ (this corresponds to $S=0$ and $\sigma=1$).  The
measurements we use for reconstructions are $f(Q(\bx;\omega))-f(Q_0)$
(for two different frequencies).

Reconstructions are again performed on the coarse grid consisting of the
nodes $(kh_c,lh_c)$ for $k,l=0,\ldots,80$ and $h_c = 1/80$.  For each
method (inverse Born series order 5, Gauss-Newton, and
Chebyshev-Halley), an approximation of the complex Schr\"odinger
potential $Q(\bx;\omega)$ is found from the frequency domain data
$f(Q(\bx;\omega))-f(Q_0)$ for $\omega=1,10$. The parameters $S$ and
$\sigma$ are then estimated with the procedure of
section~\ref{sec:twofreq}. The grid used for solving the problems
\eqref{eq:twofreq:r1} for the conductivity is the same coarse grid used
for the reconstructions (to avoid an inverse crime). The boundary conditions for
\eqref{eq:twofreq:r1} are obtained from the true conductivity evaluated
at appropriate points.

Figure~\ref{fig:ht_nonoise} shows the reconstructions of the hydraulic
conductivity $\sigma$ and storage coefficient $S$ when data has no
noise. The conductivity $\sigma$ is smooth and $|1-\sigma| < 0.8$. The
storage coefficient $S$ is also smooth and $ -5 \leq S\leq 3$. We use
the true conductivity $\sigma$ inside the wells but the storage
coefficient $S$ inside the wells is computed, as in the rest of the
domain, from \eqref{eq:Sest}.  Reconstructions with 1\% additive zero
mean Gaussian noise are included in figure~\ref{fig:ht_noise1}. As
before this means the noise has standard deviation
$0.01\max_{i,j}|[f(Q(\bx;\omega))-f(Q_0)]_{i,j}|$, which is different for the
two frequencies we use.  Similarly, figure~\ref{fig:ht_noise} displays
reconstructions with 5\% additive zero mean Gaussian noise.

\begin{remark} 
In our experiments, the parameters $\sigma$ and $S$ are chosen so that
the corresponding Schr\"odinger potential $Q(\bx;\omega)$ and the
generated data are small enough to satisfy the hypotheses of
theorem~\ref{thm:error} (for $\omega=1,10$). This makes the contrasts in
$\sigma$ (especially) and $S$ too small to represent a realistic problem
(see e.g.  \cite{Cardiff:2011:THT}). As noted before in
remark~\ref{rem:sigma_scale}, it may be possible to overcome this by
using the inverse Born series on the hydraulic tomography problem
directly.  
\end{remark}

\begin{figure}
\centering
\includegraphics[width=\textwidth]{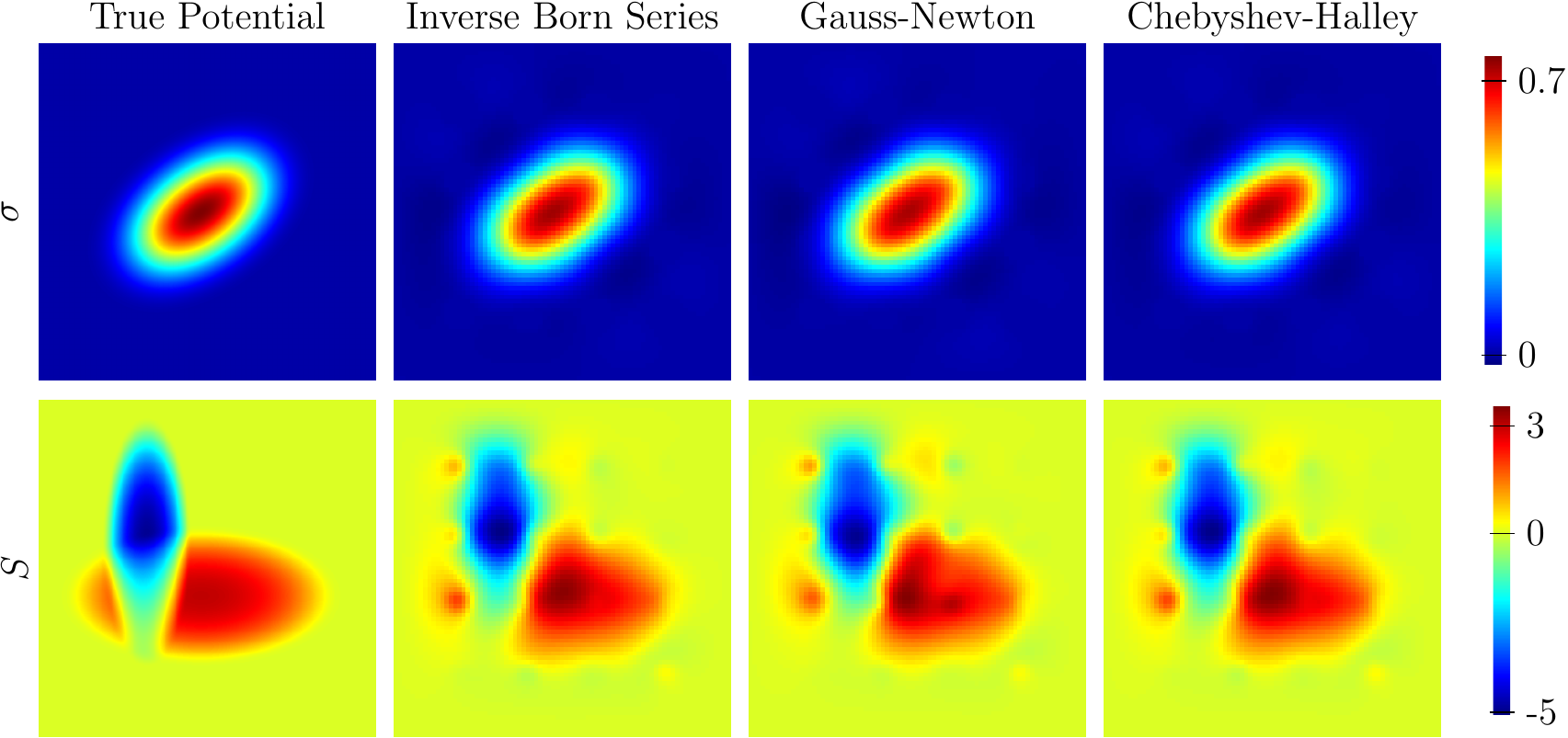}
\caption{Hydraulic tomography reconstructions of the hydraulic
conductivity $\sigma(\bx)$ (top) and the storage coefficient $S(\bx)$
(bottom) for noiseless data and different
methods.}\label{fig:ht_nonoise}
\end{figure}

\begin{figure}
\centering
\includegraphics[width=\textwidth]{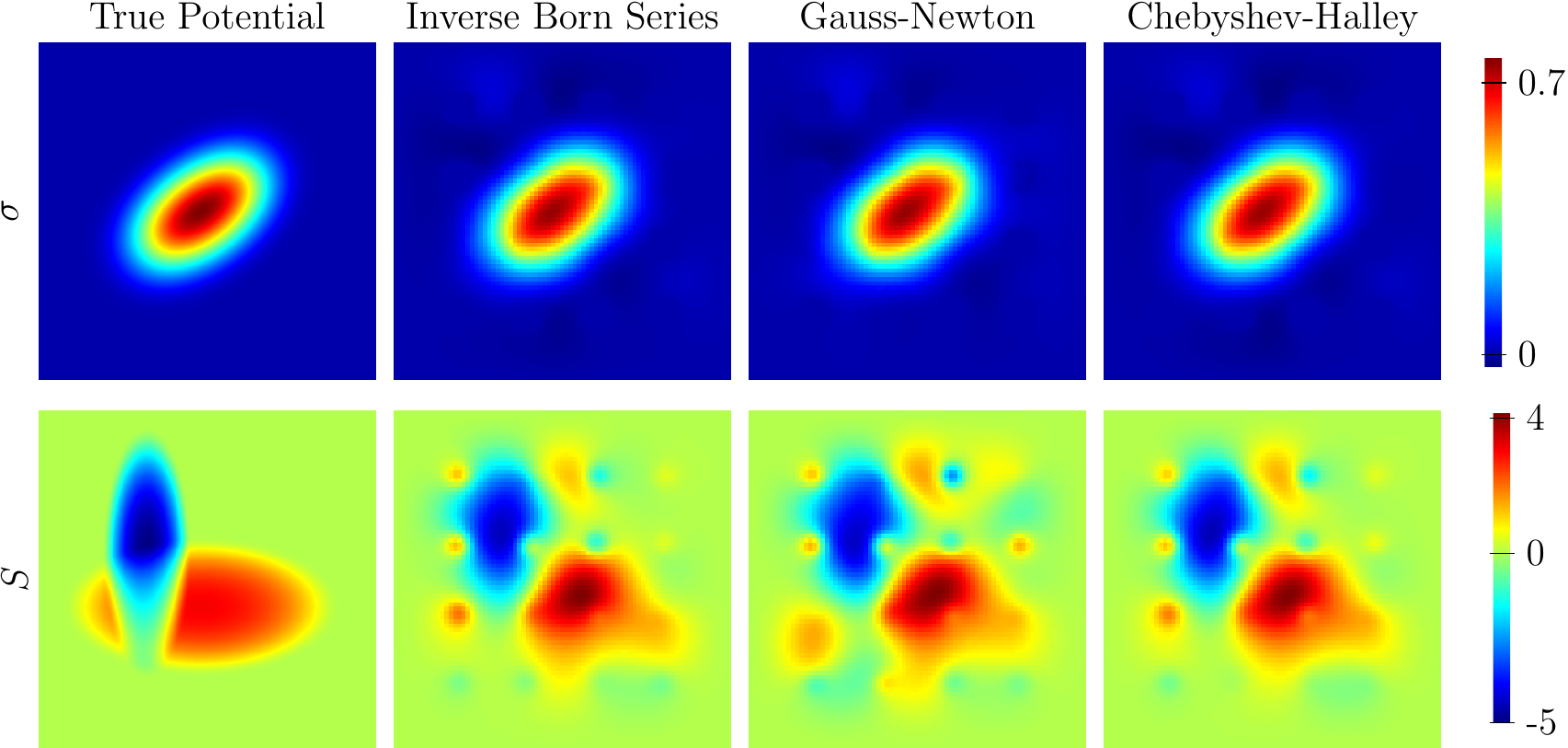}
\caption{Hydraulic tomography reconstructions of the hydraulic
conductivity $\sigma(\bx)$ (top) and the storage coefficient $S(\bx)$
(bottom) for data with 1\% additive Gaussian noise and different
methods.}\label{fig:ht_noise1}
\end{figure}

\begin{figure}[!hbtp]
\centering
\includegraphics[width=\textwidth]{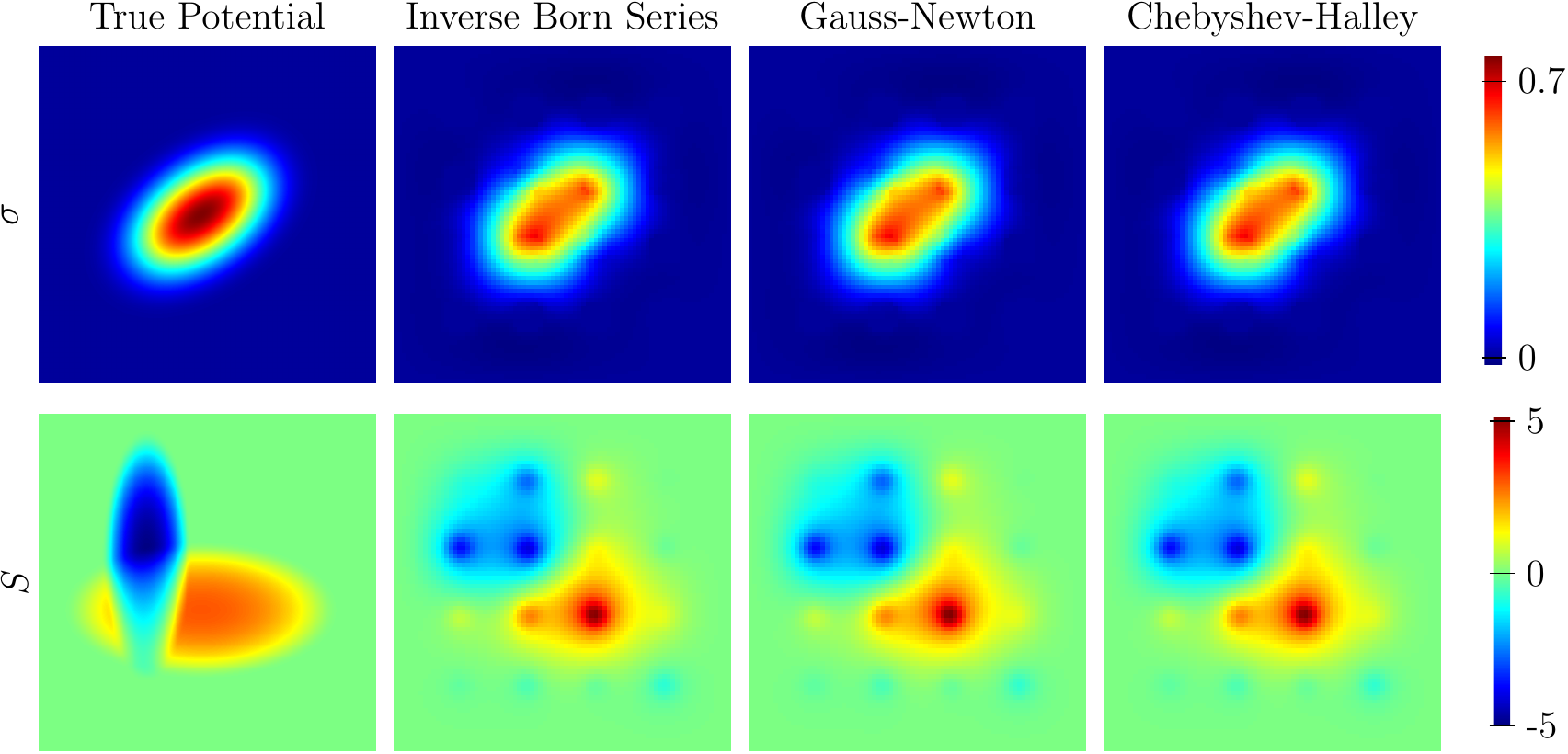}
\caption{Hydraulic tomography reconstructions of the hydraulic
conductivity $\sigma(\bx)$ (top) and the storage coefficient $S(\bx)$
(bottom) for data with 5\% additive Gaussian noise and different
methods.}\label{fig:ht_noise}
\end{figure}

\section{Discussion}\label{sec:discussion}
We show here that with little modification, the inverse Born
series convergence results of \citet*{Moskow:2008:CSI} can be
generalized to mappings between Banach spaces. With this abstraction, we
only need to show that the forward Born operators are bounded as in
\eqref{eq:fwdest} to obtain convergence, stability and error estimates
for the inverse Born series.  Such results are then proven for the
problem of finding the Schr\"odinger potential from discrete internal
measurements. A nice byproduct of our approach is that we can relate
forward and inverse Born series coefficients (up to a symmetrization) to
the Taylor series coefficients of an analytic map and its inverse
(provided it exists). 

Since the cost of computing the $n-$th term of the inverse Born series
increases exponentially in $n$, we also consider the iterative method
obtained by restarting the inverse Born series after summing the first
$k$ terms.  We obtain a class of methods that we call RIBS($k$) and that
includes the well-known Gauss-Newton and Chebyshev-Halley iterative
methods. Our numerical results show these methods give reconstructions
comparable to those obtained with the inverse Born series.

Among the future directions of this work would be to show the RIBS($k$)
method is convergent. We conjecture that the convergence rate of
RIBS($k$) is of order $k$. The RIBS($k$) method is only locally
convergent, meaning that we need to be already close to the solution for
the method to converge. Globalization strategies that keep, when
possible, this higher order convergence rate are needed. 

The application we use to illustrate our method is a problem related to
transient hydraulic tomography. Since we convert this problem to the
problem of finding a Schr\"odinger potential and all the methods we use
here are locally convergent, the contrasts that we can deal with are far
from realistic ones. We believe that a proper globalization strategy
will allow us to deal with higher contrasts. Another important question
that we have not dealt with here is that of regularization. The only
regularization that we consider here is the choice of the linear
operator that primes the inverse Born series. By analogy with what can
be done with the Gauss-Newton method, we believe it is possible to
include specific a priori information about the true parameters by
formulating the problem as minimizing the misfit plus a penalty term
that takes into account the a priori information.

\section*{Acknowledgments}
The authors would like to thank Liliana Borcea, Alexander V. Mamonov, Shari
Moskow and John Schotland for insightful conversations on this subject.
FGV is grateful to Otmar Scherzer for pointing out reference
\cite{Hettlich:2000:SMIP}. The work of the authors
was partially supported by the National Science Foundation grant
DMS-0934664.

\appendix
\section{Inverse Born series in Banach spaces}
\label{app:bornproof}
The proofs in this appendix are an adaptation of the proofs in
\citet{Moskow:2008:CSI} to inverse Born series in Banach spaces. The
results are stated in section~\ref{sec:bornconv}.

\subsection{Proof of bounds for inverse Born series coefficients
(lemma~\ref{lem:bj})}
\begin{proof}
Since $\|a_n\| \leq \alpha \mu^n$, we can estimate for $n\geq 2$:
\begin{equation}  
 \begin{aligned}
 \| b_n \| &\leq \sum_{m=1}^{n-1} \sum_{s_1+\cdots+s_m = n} \| b_m \| \| a_{s_1} \| \cdots \| a_{s_m} \| \| b_1\|^n\\
 &\leq \|b_1\|^n \sum_{m=1}^{n-1}  \| b_m \| \sum_{s_1+\cdots+s_m = n} (\alpha\mu^{s_1}) \ldots (\alpha\mu^{s_m})\\
 & = \| b_1\|^n \mu^n \sum_{m=1}^{n-1}  \| b_m \|  \alpha^m \sum_{s_1+\cdots+s_m = n} 1.
 \end{aligned}
\end{equation}
The last sum is the number of partitions of the integer $n$ into $m$ ordered parts. Hence for $n \geq 2$, we get
\begin{equation}\label{eq:binomial}
 \begin{aligned}
  \| b_n \| & \leq ( \mu \| b_1\|  )^n \sum_{m=1}^{n-1}  \| b_m \| \alpha^m { n-1 \choose m-1}\\
  & \leq  ( \mu \| b_1\|  )^n \left( \sum_{m=1}^{n-1} \| b_m \| \right)
   \left( \sum_{m=1}^{n-1} \alpha^m { n-1 \choose m-1} \right)\\
  & \leq (  \mu \| b_1\|  (\alpha+1) )^n \sum_{m=1}^{n-1} \| b_m \|.
 \end{aligned}
\end{equation}
To get the last inequality we used that 
\[
 \begin{aligned}
 \sum_{m=1}^{n-1} \alpha^m { n-1 \choose m-1} =  \sum_{m=0}^{n-2} \alpha^{m+1} { n-1 \choose m} \leq 
 \alpha \sum_{m=0}^{n-1} \alpha^{m+1} { n-1 \choose m} = \alpha (1+\alpha)^{n-1} \leq (1+\alpha)^n.
 \end{aligned}
\]
Following \cite{Moskow:2008:CSI} we can estimate the coefficients in the inverse Born series by
\begin{equation}
 \|b_n\| \leq C_n (  \mu \| b_1\|  (\alpha+1)   )^n \| b_1\|, ~\text{for}~ n \geq 2,
 \label{eq:bjbd}
\end{equation}
where the constants $C_n$ are defined recursively by
\begin{equation}
 C_2  = 1 ~ \text{and} ~
 C_{n+1} = 1 + (  (\alpha+1) \mu \| b_1\|  )^n ~ \text{for $n \geq 2$.}
\end{equation}
The constants $C_n$ are then
\begin{equation}
 C_n = \prod_{m=2}^{n-1} (1 + (  (\alpha+1)\mu \| b_1\|  )^m) 
 \leq \exp\left( \frac{1}{1 - (\alpha+1)\mu \| b_1\|} \right).
 \label{eq:cj}
\end{equation}
where the bound for $C_n$ can be derived as in \cite{Moskow:2008:CSI} and is valid when $(\alpha+1)\mu \| b_1\| < 1$, which is one of the hypothesis. The result follows from the bounds \eqref{eq:bjbd} and \eqref{eq:cj}.
\end{proof}

\subsection{Proof of local convergence of inverse Born series
(theorem~\ref{thm:smallness:inv})}
\begin{proof}
 Using the estimate of lemma~\ref{lem:bj}, we can dominate the term of the inverse Born series by a geometric series as follows
 \begin{equation}
  \begin{aligned}
  \| b_n(d^n)\| & \leq  \beta ((\alpha+1)\mu \| b_1 \| \| d \| )^n.
  \end{aligned}
 \end{equation}
 Therefore the Born series is absolutely convergent when $(\alpha+1)\mu \| b_1 \| \| d \| < 1$, which is one of the assumptions of this theorem. The tail of the series with terms the absolute values of the inverse Born series terms,  can be estimated by noticing that:
 \begin{equation}
  \sum_{N+1}^\infty \beta ((\alpha+1)\mu \| b_1 \| \| d \| )^n  = \beta  \frac{((\alpha+1)\mu \| b_1 \| \| d \|)^{N+1}}{1 - (\alpha+1)\mu \| b_1 \| \| d \|}.
 \end{equation}
\end{proof}

\subsection{Proof of stability of inverse Born series
(theorem~\ref{thm:stability})}
\begin{proof}
 We use an identity on tensor products to conclude that
 \begin{equation}
  \begin{aligned}
  \| h_1 - h_2 \| & \leq \sum_{n=1}^\infty \| b_n( d_1^{\otimes n} - d_2^{\otimes n} ) \|\\
  & = \sum_{n=1}^\infty \left\| b_n \left(\sum_{k=0}^{n-1} d_1^{\otimes k} \otimes (d_1-d_2) \otimes d_2^{\otimes(n-k-1)} \right) \right\| \\
  & \leq  \sum_{n=1}^\infty n M^{n-1}  \|b_n\| \| d_1 - d_2 \|.
  \end{aligned}
 \end{equation}
 The desired estimate follows from applying the estimate for the $\|b_n\|$ in lemma~\ref{lem:bj},
 \begin{equation}
  \begin{aligned}
    \| h_1 - h_2 \| & \leq  \| d_1 - d_2 \| \sum_{n=1}^\infty n M^{n-1} \beta ((\alpha+1)\mu \|b_1\|)^n\\
    & \leq \| d_1 - d_2 \| \frac{\beta}{M} \frac{1}{ (1 - M(\alpha+1) \mu \|b_1\|)^2 },
  \end{aligned}
 \end{equation}
 since we assumed that $M(\alpha+1) \mu \|b_1\| < 1$.
 Here we used the following inequality:
 \[
  \beta \sum_{n=1}^\infty n M^{n-1} \delta^n = \frac{\beta}{M} \sum_{n=1}^\infty n (M\delta)^n
  \leq \frac{\beta}{M}  \sum_{n=0}^\infty (n+1) (M\delta)^n = \frac{\beta}{M} \frac{1}{(1-M\delta)^2}
 \]
 where $\delta \equiv (\alpha+1) \mu \|b_1\|$.
\end{proof}

\subsection{Proof of inverse Born series error estimate
(theorem~\ref{thm:error})}
\begin{proof}
 Taking the expression for $d$ in \eqref{eq:bornconv} and replacing in
 the expression for $h_*$ in \eqref{eq:bornconv} we get:
 \begin{equation}
  h_* = \sum_{n=1}^\infty c_n (h^{\otimes n}),
 \end{equation}
 where
 \begin{equation}
  \begin{aligned}
   c_1 & = b_1 a_1,\\
   c_n &= \left( \sum_{m=1}^{n-1} b_m \left( \sum_{s_1 + \cdots s_m = n} a_{s_1} \otimes \cdots \otimes a_{s_m} \right) \right) + b_n(a_1^{\otimes n}), ~ \text{for $n \geq 2$.}
  \end{aligned}
 \end{equation}
 Using the expression \eqref{eq:iborn:coeff} of $b_n$ in terms of $b_m$, $1 \leq m \leq n-1$, we get for $n \geq 2$ that
 \begin{equation}
  c_n = \sum_{m=1}^{n-1} b_m \left( \sum_{s_1 + \cdots s_m = n} a_{s_1} \otimes \cdots \otimes a_{s_m} \right) \left( I - (b_1 a_1)^{\otimes n} \right).
 \end{equation}
 Hence the reconstruction error is
 \begin{equation}
  h - h_* = (h - b_1 a_1 h) - \sum_{n=2}^\infty \sum_{m=1}^{n-1} b_m \left( \sum_{s_1 + \cdots s_m = n} a_{s_1} \otimes \cdots \otimes a_{s_m} \right) \left( h^{\otimes n} - (b_1 a_1 h)^{\otimes n} \right).
 \end{equation}
 We now estimate the error:
 \begin{equation}
  \| h - h_* \| \leq \| h - b_1 a_1 h \| + \sum_{n=2}^\infty \sum_{m=1}^{n-1}  \sum_{s_1 + \cdots s_m = n}\|b_m\| \| a_{s_1} \| \cdots \| a_{s_m} \| \left\| h^{\otimes n} - (b_1 a_1 h)^{\otimes n} \right\|.
 \end{equation}
 For $n \geq 1$ we can estimate:
 \begin{equation}
  \begin{aligned}
  \| h^{ \otimes n} - (b_1 a_1 h)^{\otimes n}  \| & = \left\| \sum_{k=0}^{n-1} h^{\otimes k} \otimes ( h - b_1a_1 h) \otimes (b_1 a_1 h)^{\otimes (n-k-1)} \right \|\\
  & \leq n M^{n-1} \| h - b_1 a_1 h \|,
  \end{aligned}
 \end{equation}
 where we used the hypothesis $\| h \| \leq M$, $\| b_1 a_1 h \| \leq M$.
 Since we assumed the Born series coefficients satisfy $\| a_n \| \leq \alpha
 \mu^n$ we get:
 \begin{equation}
  \begin{aligned}
   \| h - h_* \| &\leq
   \| h - b_1 a_1 h \|\left( 1 + \sum_{n=2}^\infty \sum_{m=1}^{n-1} \sum_{s_1 + \cdots s_m = n}\|b_m\| (\alpha \mu^{s_1}) \cdots (\alpha \mu^{s_m}) n M^{n-1} \right)\\
   & = \| h - b_1 a_1 h \| \left( 1 + \sum_{n=2}^\infty \sum_{m=1}^{n-1} \| b_m \| \alpha^m n \mu^n M^{n-1} { n-1 \choose m-1}\right).
  \end{aligned}
 \end{equation}
 Here we have used again the fact that the number of ordered partitions of $n$ into $m$ integers is:
 \[
  \sum_{s_1 + \cdots s_m = n} 1 = { n-1 \choose m-1}.
 \]
 Clearly we have that:
 \begin{equation}
   \| h - h_* \| \leq \| h - b_1 a_1 h \| \left( 1 + \sum_{n=2}^\infty n \mu^n M^{n-1} \left( \sum_{m=1}^{n-1} \| b_m \| \right) \left( \sum_{m=1}^{n-1} \alpha^m { n-1 \choose m-1}\right)\right).
 \end{equation}
 Now using the two facts:
  \begin{equation}
  \begin{aligned}
  \sum_{m=1}^{n-1} \| b_m \| &\leq \beta \sum_{m=1}^{n-1} ((\alpha+1) \mu \| b_1 \|)^m ~ \text{(lemma~\ref{lem:bj})},\\
  \sum_{m=1}^{n-1} \alpha^m { n-1 \choose m-1 } &\leq (1 + \alpha)^n ~\text{(as in \eqref{eq:binomial})},
  \end{aligned}
 \end{equation}
we get the inequality
\begin{equation}
 \| h - h_* \| \leq \| h - b_1 a_1 h \| \left(  1 + \sum_{n=2}^\infty \frac{n}{M} (\mu M (1+\alpha))^n \beta \sum_{m=1}^{n-1} ((\alpha+1) \mu \| b_1 \|)^m\right).
\end{equation}
Adding the $m=0$ term to the geometric series over $m$ and summing we get:
\begin{equation}
 \| h - h_* \| \leq \| h - b_1 a_1 h \| \left( 1 + \frac{\beta}{M} \sum_{n=1}^\infty n (\mu M (1+\alpha) )^n \frac{ 1 - ((\alpha+1) \mu \| b_1 \|)^n}{1 -(\alpha+1) \mu \| b_1 \| } \right).
\end{equation}
The hypothesis $\mu M (\alpha + 1) <1$ and $\mu (\alpha+1) \| b_1 \| < 1$ imply the quantitity in parenthesis is bounded and depends only on $M$, $\alpha$, $\beta$ and $\mu$ and $\|b_1\|$.
\end{proof}

\bibliographystyle{abbrvnat}
\bibliography{itborn}

\end{document}